\documentclass[11pt]{article}
\usepackage[T1]{fontenc}
\usepackage{lmodern,mathrsfs,amsmath,amsfonts,amsthm,microtype}
\usepackage[a4paper,margin=27mm]{geometry}
\usepackage[hidelinks]{hyperref}
\numberwithin{equation}{section}
\newtheorem{theorem}{Theorem}[section]
\newtheorem{lemma}[theorem]{Lemma}
\newtheorem{proposition}[theorem]{Proposition}
\newtheorem{corollary}[theorem]{Corollary}
\newtheorem{conjecture}[theorem]{Conjecture}
\theoremstyle{definition}\newtheorem{definition}[theorem]{Definition}
\newcommand{\C}{\mathbb C}
\newcommand{\eps}{\varepsilon}
\newcommand{\Dx}{\partial_x}
\newcommand{\wt}{\operatorname{wt}}
\newcommand{\Can}{\operatorname{Can}}
\newcommand{\dd}{\,d}
\newcommand{\KK}{\overline K}
\title{\Large\bfseries A proof of the Hodge universality conjecture at nonzero dispersion}
\author{Francisco Hernández Iglesias\\{\small kurtos.ai}}
\date{}
\begin{document}
\maketitle
\begin{abstract}
We prove that a scalar tau-symmetric Hamiltonian deformation of the Riemann hierarchy in generalized standard form is uniquely determined by the nonzero coefficient of $u_x^2$ and the coefficients of $u_{xx}^g$, $g\ge2$, in its first Hamiltonian density. This determines the original hierarchy up to normal Miura transformations. Combined with the construction of Hodge-class double ramification hierarchies by Buryak–Rossi, this establishes Hodge universality at nonzero dispersion.
\end{abstract}

\setcounter{tocdepth}{2}
{\footnotesize\tableofcontents}
\clearpage

\section{Introduction}\label{sec:intro}
The Riemann hierarchy is the system of commuting scalar equations
\[
 \frac{\partial u}{\partial t_p}=\frac{u^p}{p!}u_x,
 \qquad p\ge0.
\]
Its first nontrivial equation is $u_{t_1}=uu_x$. Dispersive deformations add terms involving higher $x$-derivatives, with a formal parameter $\eps$ recording their differential order. The classification problem asks which such deformations remain Hamiltonian and tau-symmetric, and when two of them are related by a change of dependent variable of Miura type. This belongs to the perturbative classification of scalar integrable equations studied in \cite{LZ} and, for conservation laws beyond the Hamiltonian setting, in \cite{ALM}.

Dubrovin et al.~\cite{DLYZ} construct a family of tau-symmetric Hamiltonian deformations of the Riemann hierarchy from Hodge integrals, which are intersection numbers involving characteristic classes of the Hodge bundle on moduli spaces of stable curves. They conjecture that every nontrivial tau-symmetric Hamiltonian deformation of the Riemann hierarchy is equivalent to a member of this family under a normal Miura transformation \cite[Conjecture 1.7]{DLYZ}.

Related progress has been made under additional hypotheses. Yang--Zagier prove universality of the Witten--Kontsevich mapping hierarchy for bihamiltonian perturbations of the Riemann hierarchy possessing a tau-structure \cite[Theorem 13]{YZ}; their notion of tau-structure is weaker than the tau-symmetry used here \cite[Remark 12]{YZ}. Buryak--Xu--Yang classify tau-symmetric Hamiltonian deformations at nonzero dispersion that admit a compatible second Poisson structure, showing that they are Miura-equivalent to the special-Hodge hierarchy \cite[Theorem 4.2]{BXY}. For the broader class of deformations possessing a tau-structure, without assuming a polynomial Hamiltonian structure, Liu et al. construct further families and formulate classification and universality conjectures \cite[Conjectures 5.6--5.7]{LRYZ}.

Buryak--Rossi \cite{BR} approach this classification problem using the geometry of double ramification (DR) cycles \cite{Bur, BDGR, BDGRtype}. After matching the dispersion normalization, the Hodge-class DR hierarchies are normally Miura-equivalent to the Hodge hierarchies of \cite{DLYZ}, by the strong DR/DZ equivalence theorem~\cite[Theorem 1.2]{BLS}, applied to rank-one Hodge cohomological field theories. Conjecture~3.7 in~\cite{BR} restates the standard-form conjecture \cite[Conjecture 6.1]{DLYZ} in terms of the first Hamiltonian. We recall its three assertions to specify the part addressed here. The terms \emph{normal Miura transformation} and \emph{standard form} are defined precisely in Section~\ref{sec:notation}.

\begin{conjecture}[{\cite[Conjecture 3.7]{BR}}]\label{conj:DLYZ}
Every scalar tau-symmetric Hamiltonian deformation of the Riemann hierarchy satisfies the following assertions.
\begin{enumerate}
 \item A unique normal Miura transformation puts its Poisson operator and first Hamiltonian in standard form:
 \[
 K=\partial_x,\qquad
 \overline H_1=\int\left(\frac{u^3}{6}+a\eps^2u_1^2
       +\sum_{g\ge2}\eps^{2g}H_{1,2g}\right)\dd x.
 \]
 Here $u_j=\partial_x^j u$, and each $H_{1,2g}$ is a constant-coefficient polynomial of differential weight $2g$, with every derivative index at least two and the largest index in each monomial occurring at least twice.
 \item If $a=0$, then every $H_{1,2g}$ vanishes.
 \item If $a\ne0$, all coefficients are determined by $a$ and the coefficients $A_g$ of $u_2^g$ in $H_{1,2g}$, for $g\ge2$.
\end{enumerate}
\end{conjecture}

Buryak--Rossi prove the existence and uniqueness of a \emph{generalized} standard form, which may contain odd powers of $\eps$ \cite[Theorem 3.8]{BR}. They also construct Hodge-class DR hierarchies in standard form with arbitrary $a\ne0$ and arbitrary diagonal data $A_g$ \cite[Theorem 3.9]{BR}. The remaining classification question on this locus is uniqueness from these data. Once uniqueness is established in generalized standard form, comparison with their geometric family also forces all odd terms to vanish.

This paper proves the uniqueness theorem below. This proves part~(3) of Conjecture~\ref{conj:DLYZ} and gives part~(1) on the locus $a\ne0$ as a consequence. Hodge universality and evenness follow by combining it with the results of \cite{BR} (Corollary~\ref{cor:universality}).

\begin{theorem}\label{thm:generic}
Two tau-symmetric Hamiltonian deformations of the Riemann hierarchy in generalized standard form with the same $a\ne0$ and the same diagonal coefficients $A_g$, $g\ge2$, have identical Hamiltonians.
\end{theorem}
The theorem is proved in Section~\ref{sec:nonlinear}. As an immediate consequence, one obtains

\begin{corollary}[Hodge universality at nonzero dispersion]\label{cor:universality}
Every tau-symmetric Hamiltonian deformation of the Riemann hierarchy with $a\ne0$ is normally Miura-equivalent to the unique Hodge-class DR hierarchy with the same leading and diagonal coefficients. Its generalized standard representative is even in $\eps$, and hence is in standard form.
\end{corollary}
\begin{proof}
By \cite[Theorem 3.8]{BR}, put the given hierarchy in generalized standard form. By \cite[Theorem 3.9]{BR}, there is a Hodge-class DR hierarchy in even standard form with the same $a$ and $A_g$; its geometric parameters are uniquely determined by the triangular formula in that theorem. Theorem~\ref{thm:generic} identifies these two representatives.
\end{proof}

To prove Theorem~\ref{thm:generic}, we compare two hierarchies in generalized standard form: after a common normalization of $a$, their first discrepancy satisfies a homogeneous polynomial equation obtained by linearizing Poisson commutativity and specialness at the first two KdV Hamiltonians. The kernel theorem, Theorem~\ref{thm:pair}, shows that this discrepancy is determined by its diagonal coefficient, so equality of the diagonal data forces the two hierarchies to agree, as proved in Section~\ref{sec:nonlinear}.

Part~(2), concerning $a=0$, will be dealt with in a separate dedicated paper.

\section{Preliminaries}\label{sec:notation}
\subsection{Differential polynomials and local functionals}
We work over $\C$ and write
\[
 u_0=u,\qquad u_j=\partial_x^j u\ (j\ge1),\qquad
 \mathcal A_u=\C[[u]][u_1,u_2,\ldots].
\]
A differential polynomial involves only finitely many positive-order jets. The total derivative and the partial derivative with respect to the undifferentiated variable are
\[
 \partial_x=\sum_{j\ge0}u_{j+1}\frac{\partial}{\partial u_j},
 \qquad \partial_u=\frac{\partial}{\partial u_0}.
\]
In $\partial_u$, all $u_j$ with $j\ge1$ are held fixed. This operator is different from the variational derivative defined below.

A local functional is a class
\[
 \overline A=\int A\dd x\in
 \Lambda_u:=\mathcal A_u/(\partial_x\mathcal A_u+\C).
\]
The integral is formal: densities that differ by a total derivative or a constant define the same functional.
We use the variational derivative and the bracket associated with the Poisson operator $\partial_x$:
\begin{equation}\label{eq:variational}
 \frac{\delta\overline A}{\delta u}
 =\sum_{j\ge0}(-\partial_x)^j\frac{\partial A}{\partial u_j},
 \qquad
 \{\overline A,\overline B\}
 =\int\frac{\delta\overline A}{\delta u}\,
       \partial_x\frac{\delta\overline B}{\delta u}\dd x.
\end{equation}

Assign differential weight $j$ to $u_j$ and weight zero to $u$. A homogeneous Hamiltonian density is a formal series
\begin{equation}\label{eq:epshom}
 H_p=\sum_{N\ge0}\eps^N H_{p,N},\qquad
 \deg_{\partial}H_{p,N}=N.
\end{equation}
Equivalently, $\deg_{\partial}\eps=-1$ and $H_p$ has total differential degree zero.

\subsection{Hamiltonian, tau-symmetric, and special hierarchies}

A Hamiltonian deformation of the Riemann hierarchy consists of pairwise commuting flows
\[
 \frac{\partial u}{\partial t_p}
   =K\frac{\delta\overline H_p}{\delta u},\qquad
 \overline H_p=\int\frac{u^{p+2}}{(p+2)!}\dd x+O(\eps),
 \quad p\ge0,
\]
where $K=\partial_x+O(\eps)$ is a Poisson operator of differential degree one. Tau-symmetry means that one can choose densities $h_p=u^{p+2}/(p+2)!+O(\eps)$ of $\overline H_p$, together with $h_{-1}=u$, such that
\[
 \frac{\partial h_{p-1}}{\partial t_q}
 =\frac{\partial h_{q-1}}{\partial t_p},\qquad p,q\ge0.
\]
For the general operator $K$, this definition also includes
$K(1)=0$ and $K\,\delta\overline H_0/\delta u=u_x$.

A Miura transformation is a formal substitution
\[
 \widetilde u=u+\sum_{N\ge1}\eps^N f_N(u,u_1,\ldots),
 \qquad \deg_{\partial}f_N=N.
\]
It is \emph{normal} if it can be written as
\[
 \widetilde u=u+\partial_x^2P,\qquad
 P=\sum_{N\ge2}\eps^N P_N,\quad
 \deg_{\partial}P_N=N-2.
\]

By \cite[Theorem 3.8]{BR}, every tau-symmetric Hamiltonian deformation admits a unique normal Miura transformation putting it in \emph{generalized standard form}:

\begin{definition}[Generalized standard form]\label{def:standard}
A tau-symmetric Hamiltonian deformation is in \emph{generalized standard form} if its Poisson operator is $\partial_x$ and its first Hamiltonian is represented by the density
\begin{equation}\label{eq:standard}
 H_1=\frac{u^3}{6}+a\eps^2u_1^2
       +\sum_{N\ge4}\eps^N H_{1,N},
 \qquad a\in\C.
\end{equation}
Each $H_{1,N}$ is a constant-coefficient polynomial of differential weight $N$. Every monomial involves only $u_j$ with $j\ge2$, and its largest derivative index occurs at least twice. The hierarchy is in \emph{standard form} if $H_{1,N}=0$ for every odd $N$. For $g\ge2$, the diagonal coefficient is
\[
 A_g:=[\eps^{2g}u_2^g]H_1=[u_2^g]H_{1,2g},
\]
where square brackets denote coefficient extraction in the displayed density.
\end{definition}

Normal Miura transformations preserve tau-symmetry \cite[Lemma 4.3]{DLYZ}. By \cite[Remark 3.3]{BR}, for $K=\partial_x$ tau-symmetry is equivalent to \emph{specialness}, that is, $\partial_u\overline H_1=\overline H_0$, with $\overline H_0=\int u^2/2\dd x$. In these coordinates the two normalization identities above hold automatically. Specialness implies
\begin{equation}\label{eq:special}
 \partial_u\overline H_p=\overline H_{p-1},\qquad p\ge1.
\end{equation}
From now on, we therefore work with special Hamiltonian hierarchies in generalized standard form. Odd powers of $\eps$ are still allowed at this stage.

\subsection{Canonical densities and partition notation}
For a partition $\lambda=(\lambda_1\ge\cdots\ge\lambda_l\ge1)$, set
\[
 |\lambda|=\sum_{i=1}^l\lambda_i,\qquad
 l(\lambda)=l,\qquad u_\lambda=\prod_{i=1}^l u_{\lambda_i}.
\]
We use the partition sets of \cite{BR}:
\begin{align*}
 \mathcal P_N^\circ
 &=\{\lambda:|\lambda|=N,\ l(\lambda)\ge2,\ \lambda_1=\lambda_2\},\\
 \mathcal P_N'
 &=\{\lambda\in\mathcal P_N^\circ:\lambda_l\ge2\}.
\end{align*}
A density is \emph{canonical} if its positive-weight terms are linear combinations of $u_\lambda$ with $\lambda\in\mathcal P_N^\circ$, with coefficients in $\C[[u]]$. Its weight-zero term is chosen to vanish at $u=0$. This specifies a unique representative of each local functional \cite[Lemma 2.3]{BR}. The condition $\lambda_1=\lambda_2$ says exactly that the highest derivative appears at least twice.

A \emph{derivative-only} density has constant coefficients and contains no $u=u_0$.

In this notation, the corrections in Definition~\ref{def:standard} are
\[
 H_{1,N}=\sum_{\lambda\in\mathcal P_N'}a_\lambda u_\lambda,
 \qquad a_\lambda\in\C.
\]
In particular, $A_g=a_{(2^g)}$, where $(2^g)$ is the partition with $g$ entries equal to two.
The distinction between the two forms is only the possible presence of odd powers of $\eps$. Each higher term in \eqref{eq:standard} contains at least two jet factors, and none contains $u_1$. For example, the weight-six possibilities are $u_3^2$ and $u_2^3$; only the latter is diagonal.

The main result, Theorem~\ref{thm:generic}, is uniqueness from these coefficients.

\subsection{Normalization and gradings}\label{sec:normalization}
For $a\ne0$, rescale the dispersion parameter by $\widehat\eps=\sqrt{-a}\,\eps$, using the same choice for both hierarchies, and relabel it as $\eps$; henceforth $a=-1$. We suppress powers of $\eps$, recovering them by multiplying each monomial of differential weight $N$ by $\eps^N$.

Assign the KdV weight $\wt(u_j)=j+2$, and grade a monomial of length $L$ in $H_p$ by the coefficient degree $\mu=N+2L-2(p+2)$, counting undifferentiated $u$'s among the factors. For fixed $p,\mu$, the identity
\[
 N+2L=2(p+2)+\mu
\]
bounds both $N$ and $L$, and hence leaves only finitely many monomials. Thus each coefficient-degree component is a polynomial, even when the full density has formal power-series coefficients in $u$. We work componentwise in this grading, with the KdV background in degree zero.

\section{The KdV pair kernel theorem}\label{sec:finite}
With the normalization and suppression convention of Section~\ref{sec:normalization}, the first two KdV densities are
\begin{equation}\label{eq:base}
 H=\frac{u^3}{6}-u_1^2,\qquad
 J=\frac{u^4}{24}-u u_1^2+\frac65u_2^2.
\end{equation}
They satisfy $\{\overline H,\overline J\}=0$ and
$\partial_u\overline J=\overline H$. Their KdV weights are six and eight.
We call a local functional $\overline K$ a \emph{KdV centralizer} if $\{\overline K,\overline H\}=0$.

Consider infinitesimal variations $\overline H+\tau\overline F$ and
$\overline J+\tau\overline\eta$, where $\tau^2=0$.
The coefficient of $\tau$ in their commutativity equation is
$\{\overline H,\overline\eta\}+\{\overline F,\overline J\}=0$.
The linearized specialness relation is $\partial_u\overline\eta=\overline F$.
For a specified derivative-only canonical density $F$, this allows us to write
$\overline\eta=\overline{uF+G}$ with $G$ derivative-only at the weights considered here.
If $F$ has weight $d$, then $uF+G$ has weight $d+2$.

\begin{samepage}
\begin{definition}[The polynomial pair space]\label{def:pairspace}
For an integer $d\ge7$, define the following spaces of canonical \emph{densities}:
\begin{align*}
 V_d&=\operatorname{span}_{\C}
 \{u_\lambda:\lambda\in\mathcal P_N',\ N+2l(\lambda)=d\},\\
 W_{d+2}&=\operatorname{span}_{\C}
 \{u_\lambda:\lambda\in\mathcal P_N^\circ,\ N+2l(\lambda)=d+2\}.
\end{align*}
In each set, $N$ runs over nonnegative integers. Thus $V_d$ permits indices at least two, whereas $W_{d+2}$ also permits index one. Both require at least two factors and a repeated largest index. For example,
\[
 V_8=\C u_2^2,\qquad W_{10}=\C u_3^2.
\]
The pair kernel $\mathscr K_d$ consists of all $(F,G)\in V_d\oplus W_{d+2}$ such that
\begin{equation}\label{eq:pair}
 \{\overline H,\overline{uF+G}\}+\{\overline F,\overline J\}=0
 \quad\text{in }\Lambda_u.
\end{equation}
Equivalently, the bracket density
\[
 \frac{\delta\overline H}{\delta u}\,
 \partial_x\frac{\delta}{\delta u}(\overline{uF+G})
 +\frac{\delta\overline F}{\delta u}\,
 \partial_x\frac{\delta\overline J}{\delta u}
\]
is $\partial_x$-exact. This is a finite system of linear equations in the coefficients of $F$ and $G$.
\end{definition}
\end{samepage}
The multiplication $uF$ is performed on the specified canonical density. It is not multiplication of an arbitrary representative of the functional $\overline F$: replacing $F$ by $F+\partial_xQ$ changes $uF$ by a term congruent to $-u_1Q$ modulo total derivatives. Fixing the representative removes this ambiguity. In the smallest case, $d=8$, substituting $F=c u_2^2$ and $G=\beta u_3^2$ gives $\beta=-16c/7$, so $\mathscr K_8=\C\,(u_2^2,-16u_3^2/7)$.

\begin{theorem}\label{thm:pair}
For every integer $d\ge8$, the projection $(F,G)\mapsto F$ is injective on $\mathscr K_d$, and
\[
 \dim\mathscr K_d=
 \begin{cases}
 1,&d=4g,\quad g\ge2,\\
 0,&d\not\equiv0\pmod4.
 \end{cases}
\]
In the one-dimensional case its first component has a nonzero coefficient of $u_2^g$.
\end{theorem}
Theorem~\ref{thm:pair}, including existence and diagonal nonvanishing, is proved in Section~\ref{sec:pair-proof}, using Lemma~\ref{lem:topv} and the polynomial KdV hierarchy. An explicit local construction and the nonzero ratio of its diagonal and quadratic coefficients are given in Subsection~\ref{sec:diagonal}.

\section{Derivative-only KdV brackets}\label{sec:canonical}
For a derivative-only monomial $M$, let $N$ and $L$ be as in Section~\ref{sec:normalization}, and let $r$ be the exponent of $u_1$ in $M$.

\begin{lemma}\label{lem:topv}
Let $K\ne0$ be a derivative-only canonical density, homogeneous of KdV weight $k \geq 6$. Then:
\begin{enumerate}
    \item The canonical density $B$ of $\{\KK,\overline H\}$ has KdV weight $k+3$ and contains a nonzero term involving $u_1$. In particular, $B\notin V_{k+3}$.
    \item Every derivative-only KdV centralizer of positive KdV weight is zero.
\end{enumerate}

\end{lemma}
\begin{proof}
Decompose the Hamiltonian into its Riemann and Airy parts:
\[
 \overline H=\int\frac{u^3}{6}\dd x+\int(-u_1^2)\dd x.
\]
By the variational formula and integration by parts,
\[
 \{\KK,\overline H\}
 =\int\sum_{j\ge1}\frac{\partial K}{\partial u_j}
       \partial_x^j(uu_1+2u_3)\dd x.
\]
Subtracting the total derivative $\partial_x(uK)$ from the Riemann contribution gives the derivative-only representative
\begin{align*}
 R_K
 &=\sum_{j\ge1}\frac{\partial K}{\partial u_j}
       \bigl(\partial_x^j(uu_1)-u u_{j+1}\bigr)-u_1K\\
 &=u_1^2\frac{\partial K}{\partial u_1}
   +u_1\sum_{j\ge2}(j+1)u_j\frac{\partial K}{\partial u_j}-u_1K
   +\sum_{j\ge3}\frac{\partial K}{\partial u_j}
       \sum_{a=2}^{j-1}\binom ja u_a u_{j+1-a}.
\end{align*}
The Airy contribution has representative
\[
 A_K=2\sum_{j\ge1}u_{j+3}\frac{\partial K}{\partial u_j}.
\]
Neither $R_K$ nor $A_K$ is necessarily canonical. By definition, $B$ is the canonical representative of $R_K+A_K$.

We first check that canonical reduction does not increase the exponent of $u_1$ in a derivative-only polynomial. A monomial with a unique highest jet $u_m$, $m\ge2$, can be written as $u_m u_{m-1}^sQ$, with $Q$ involving only $u_1,\ldots,u_{m-2}$. Integration by parts gives
\[
 u_m u_{m-1}^sQ
 \equiv-\frac{u_{m-1}^{s+1}}{s+1}\partial_xQ
 \pmod{\partial_x\mathcal A_u}.
\]
For $m\ge3$, the prefactor introduces no $u_1$, and differentiating $Q$ cannot increase its $u_1$-exponent. For $m=2$, $Q$ is constant and the monomial is exact. A linear term $u_1$ is also exact. Repeating this reduction proves the assertion.

We use the Riemann leading-term calculation of \cite[equation (6.3)]{DLYZ}. Now write $K=\sum_M c_M M$, and let $r$ be the largest $u_1$-exponent among its nonzero monomials. For $M=\prod_{j\ge1}u_j^{m_j}$, write
\[
 N(M)=\sum_{j\ge1}j m_j,\qquad L(M)=\sum_{j\ge1}m_j.
\]
If $m_1=r$, the first three terms of $R_K$ send $M$ to
\[
 \left(r+\sum_{j\ge2}(j+1)m_j-1\right)u_1M
 =\bigl(N(M)+L(M)-r-1\bigr)u_1M.
\]
The last sum in $R_K$ involves only new factors of index at least two, so it has $u_1$-degree at most $r$. The same holds for $A_K$. Monomials of $K$ with $u_1$-exponent below $r$ also contribute only up to degree $r$. By the reduction just proved, these bounds remain valid after integration by parts. Moreover, each $u_1M$ in the displayed leading term is already canonical, since multiplication by $u_1$ preserves the repeated-highest-index condition. Therefore
\[
 B=\sum_{\substack{M:\,m_1=r}}c_M
       \bigl(N(M)+L(M)-r-1\bigr)u_1M+B_{\le r},
 \qquad \deg_{u_1}B_{\le r}\le r.
\]
A nonzero derivative-only canonical monomial has $N(M)\ge L(M)\ge2$, so
\[
 N(M)+L(M)-r-1\ge N(M)-1>0.
\]
The monomials $u_1M$ are distinct. Hence the leading sum cannot vanish, and $B$ has a nonzero term of $u_1$-degree $r+1$.

Finally, for the second part, note that a derivative-only centralizer would have $B=0$, which is impossible for nonzero canonical $K$ by the argument above.
\end{proof}

\section{Proof of Theorem~\ref{thm:pair}}\label{sec:pair-proof}

We split the proof of the theorem into five steps:
\begin{enumerate}
    \item Prove injectivity of the projection $(F, G) \in  \mathscr K_d \rightarrow F \in V_d$.
    \item Prove $\dim \mathscr K_d = 0$ for odd $d$.
    \item Prove $\dim \mathscr K_d \leq 1$ for even $d$.
    \item Prove $\dim \mathscr K_d = 0$ for $d \equiv 2 \pmod 4$.
    \item Prove $\dim \mathscr K_{4g} = 1$ by constructing a pair in every $\mathscr K_{4g}$ with nonzero quadratic and diagonal coefficients.
\end{enumerate}
Step 1 follows directly from Lemma~\ref{lem:topv}. After introducing symmetric momentum symbols in Subsection~\ref{sec:symbols}, Proposition~\ref{prop:quad} proves Steps 2 and 3 together. Proposition~\ref{prop:quartic-vanishing} proves Step 4 by deriving a necessary quartic equation and showing that it has no polynomial solution. Proposition~\ref{prop:diagonal-algebraic} supplies the construction and diagonal nonvanishing required for Step 5.

\paragraph{Step 1: injectivity of the first component.}
If $(F,G)\in\mathscr K_d$ has $F=0$, then \eqref{eq:pair} reduces to
\[
 \{\overline H,\overline G\}=0.
\]
Since $G$ is derivative-only of positive KdV weight $d+2$, Lemma~\ref{lem:topv} gives $\overline G=0$. Its canonical representative is therefore $G=0$, proving injectivity.

\subsection{Symmetric momentum symbols}\label{sec:symbols}
We use the standard symbolic representation of differential polynomials, originating in the work of Gel'fand--Dikii \cite{GD}; see \cite{MNW} for the symmetrization convention used here. The variables $\xi_1,\ldots,\xi_l$ are formal momenta, one for each factor of a monomial of length $l$.

For a polynomial density homogeneous of length $l$, define its symbol by linear extension of
\[
 \sigma_l(u_{i_1}\cdots u_{i_l})
 =\frac1{l!}\sum_{\pi\in S_l}
       \xi_{\pi(1)}^{i_1}\cdots\xi_{\pi(l)}^{i_l},
 \qquad i_1,\ldots,i_l\ge0.
\]
Here $S_l$ is the permutation group, and $i_j=0$ represents an undifferentiated factor $u$. Permuting the factors of a differential monomial does not change its symbol. Moreover, these averaged monomials form a basis of $\C[\xi_1,\ldots,\xi_l]^{S_l}$, so $\sigma_l$ is a linear isomorphism. Its inverse expands a symmetric polynomial into ordinary monomials and replaces each $\xi_1^{i_1}\cdots\xi_l^{i_l}$ by $u_{i_1}\cdots u_{i_l}$, retaining its coefficient. This is a correspondence at fixed length, rather than a substitution for the individual jet variables. Products of densities correspond to symmetrized products in the combined set of momenta
\[
\sigma_{l+k}(PQ)
=\operatorname{Sym}_{l+k}\!\left[
\sigma_l(P)(\xi_1,\ldots,\xi_l)\,
\sigma_k(Q)(\xi_{l+1},\ldots,\xi_{l+k})
\right],
\]
where $\operatorname{Sym}_{l+k}$ denotes averaging over all permutations
of $\xi_1,\ldots,\xi_{l+k}$.

The symbol of a length-$l$ monomial of differential weight $N$ has ordinary polynomial degree $N$; its KdV weight is $N+2l$. We apply the correspondence separately to each length component, and subscripts such as $F_l$ and $\eta_l$ below denote the symbols of those components. In particular, the number of momentum variables varies with $l$.

To pass from densities to local functionals, observe that the product rule gives
\[
 \sigma_l(\partial_x P)
 =(\xi_1+\cdots+\xi_l)\sigma_l(P).
\]
Define the elementary symmetric polynomials in the current $l$ momenta by
\[
 e_j=\sum_{1\le a_1<\cdots<a_j\le l}\xi_{a_1}\cdots\xi_{a_j},
 \qquad e_1=\xi_1+\cdots+\xi_l,\qquad e_l=\xi_1\cdots\xi_l.
\]
Since $\sigma_l$ is an isomorphism, a density is a total derivative if and only if its symbol is divisible by $e_1$. Thus, at fixed positive length, symbols of local functionals are represented in the quotient
\[
 \C[\xi_1,\ldots,\xi_l]^{S_l}/(e_1)
 =\C[e_1,\ldots,e_l]/(e_1)
 \cong\C[e_2,\ldots,e_l].
\]

A derivative-only monomial has every $i_j\ge1$, so its symbol is divisible by $e_l$. Conversely, every symmetric polynomial divisible by $e_l$ comes from a derivative-only density. Similarly, requiring every $i_j\ge2$ gives divisibility by $e_l^2$. After passage to the quotient, these subspaces are the ideals $(e_l)$ and $(e_l^2)$ in the functional symbol ring.

To express the specialness relation in symbols, we also need ordinary differentiation with respect to \(u\). This removes an undifferentiated factor, and consequently, for $l\ge2$,
\[
 \sigma_{l-1}(\partial_u P)
 =l\,\sigma_l(P)(\xi_1,\ldots,\xi_{l-1},0).
\]
The same formula holds for local functionals modulo the sum of the remaining momenta.

For densities $f,g$ of lengths $l,k\ge2$, write $f_l=\sigma_l(f)$ and $g_k=\sigma_k(g)$. Put $n=l+k-2$ and $s_I=\sum_{i\in I}\xi_i$ for $I\subset\{1,\ldots,n\}$. Write $\xi_I$ for the tuple of variables indexed by $I$, and $I^c$ for its complement. The symbol of $\{\overline f,\overline g\}$ is
\begin{equation}\label{eq:symbracket}
 \frac{lk}{\binom n{l-1}}
 \sum_{|I|=l-1}(-s_I)
 f_l(\xi_I,-s_I)\,g_k(\xi_{I^c},s_I),\qquad e_1=0.
\end{equation}
The factors $l$ and $k$ count the choices in the two variational derivatives. Integration by parts replaces the momentum of each removed factor by minus the sum of the momenta left in that density; the final $\partial_x$ contributes $-s_I$. Averaging over the $\binom n{l-1}$ choices of $I$ gives \eqref{eq:symbracket}.

\subsubsection{The quadratic coefficients}
To simplify the pair equation, we use the following transformations. An infinitesimal local canonical change generated by a derivative-only $K$ acts on a pair of variations by
\[
\overline F'=\overline F+\{\overline K,\overline H\},
\qquad
\overline\eta'=\overline\eta+\{\overline K,\overline J\}.
\]
The Jacobi identity gives
\begin{align*}
\{\overline H,\overline\eta'\}
+\{\overline F',\overline J\}
&=\{\overline H,\overline\eta\}
  +\{\overline F,\overline J\}\\
&\quad+\{\overline H,\{\overline K,\overline J\}\}
  +\{\{\overline K,\overline H\},\overline J\}\\
&=\{\overline K,\{\overline H,\overline J\}\}=0,
\end{align*}
where we used \eqref{eq:pair} and
$\{\overline H,\overline J\}=0$.
The functional $\int u^2/2\dd x$ generates spatial translation, so its bracket with any local functional vanishes. Together with $\partial_uK=0$, this shows that $\overline F'$ remains derivative-only as a functional and that $\partial_u\overline\eta'=\overline F'$ is preserved.

Suppose $d=2m+4$ is even. Here \emph{quadratic} means length two in the jets, not degree two in the variation parameter. By canonicality, the quadratic parts of $F$ and $\eta$ are respectively $c u_m^2$ and $\beta u_{m+1}^2$; the term $uF$ has length at least three.

Write $F_3$ for the symbol of the length-three component of $F$. Every monomial in this component has the form $u_i u_j u_k$ with $i,j,k\ge2$, so every term of its symbol is divisible by $\xi_1^2\xi_2^2\xi_3^2=e_3^2$. After setting $e_1=0$, we can therefore write
\[
 F_3=e_3^2Q(e_2,e_3).
\]
Here and below $p_j=\sum_{a=1}^l\xi_a^j$ denotes the power sum in the current $l$ momenta. In three variables, Newton's identities give $p_3=3e_3$.

Choose a derivative-only density $K$ of length three whose symbol is
\[
 K_3=-\frac{F_3}{2p_3}=-\frac{e_3Q}{6}.
\]
Such a density exists because this is a polynomial divisible by $e_3$. The quadratic part $-u_1^2$ of $H$ generates the Airy flow $u_t=2u_3$, so the length-three symbol of $\{\overline K,\overline H\}$ is $2p_3K_3$. The cubic part $u^3/6$ contributes only in length four. Hence the transformed first component satisfies
\[
 F_3'=F_3+2p_3K_3=0.
\]
A bracket of densities of lengths $a,b$ has length $a+b-2$. Since $K$ has length three and $H,J$ have no terms of length below two, this transformation changes neither quadratic coefficient $c$ nor $\beta$. Denote the transformed cubic second-component symbol by $U_3:=\eta_3'$.

For any length $l$, the quadratic background densities $-u_1^2$ and $\frac65u_2^2$ contribute to the commutativity equation
\begin{equation}\label{eq:airy}
 -2p_3\eta_l+\frac{12}{5}p_5F_l.
\end{equation}
Note the equation above holds for both pairs $(F, \eta)$ and $(F', \eta')$. We apply it to the latter, and obtain the symbol for the transformed pair at length three, so \eqref{eq:airy} becomes $-2p_3U_3=-6e_3U_3$. The remaining length-three terms pair the cubic background densities $\frac16 u^3$ and $-uu_1^2$ with the quadratic variations $\beta u_{m+1}^2$ and $c u_m^2$, respectively. By \eqref{eq:symbracket}, the contribution reads
\begin{align*}
-\frac{(-1)^m}{3} \beta p_{2m+3} + \frac{2(-1)^m }{3}c e_2p_{2m+1}.
\end{align*}
Thus commutativity at length three gives
\[
 -6e_3U_3+\frac{(-1)^m}{3}
 \bigl(-\beta p_{2m+3}+2c e_2p_{2m+1}\bigr)=0.
\]
Newton's recurrence in three variables with $e_1=0$,
\[
 p_n=-e_2p_{n-2}+e_3p_{n-3}\quad(n\ge4),
 \qquad p_2=-2e_2,\quad p_3=3e_3,
\]
shows that every odd power sum $p_{2r+1}$, $r\ge1$, is divisible by $e_3$. Therefore, we can solve the equation for a polynomial $U_3$:
\begin{equation}\label{eq:U}
 U_3=\frac{(-1)^m}{18e_3}
 \bigl(-\beta p_{2m+3}+2c e_2p_{2m+1}\bigr).
\end{equation}

We next use specialness to determine $\beta$ in terms of $c$. The relation $\partial_u\overline\eta'=\overline F'$ gives
\[
 3U_3(\xi,-\xi,0)=c(-1)^m\xi^{2m}.
\]
To evaluate \eqref{eq:U}, divide by $e_3$ as a polynomial before making this substitution. The same recurrence gives
\[
 \left.\frac{p_{2r+1}}{e_3}\right|_{(\xi,-\xi,0)}
 =(2r+1)\xi^{2r-2},\qquad r\ge1.
\]
Together with $e_2(\xi,-\xi,0)=-\xi^2$, this yields
\begin{equation}\label{eq:beta}
 \beta=-\frac{4(m+2)}{2m+3}c.
\end{equation}
In particular, $c=0$ forces $\beta=0$ and $U_3=0$.

For $c\ne0$, scale the pair so that $c=1$. Substituting \eqref{eq:beta} into \eqref{eq:U} gives
\[
 U_3=\frac{(-1)^m}{9e_3}
 \left(\frac{2m+4}{2m+3}p_{2m+3}+e_2p_{2m+1}\right).
\]
Applying Newton's recurrence
$e_2p_{2m+1}=e_3p_{2m}-p_{2m+3}$, this simplifies to
\begin{equation}\label{eq:Usimple}
 U_3=\frac{(-1)^m}{9}
 \left(p_{2m}+\frac{p_{2m+3}}{(2m+3)e_3}\right).
\end{equation}

For odd $d$, neither variation has a quadratic functional part, because every quadratic canonical monomial $u_j^2$ has even KdV weight $2j+4$. The same cubic elimination applies; with both quadratic parts zero, the cubic pair equation reduces to $-6e_3U_3=0$, so $U_3=0$ as well.

We summarize the results of this section in a lemma:

\begin{lemma}\label{lem:cubic-reduction}
Let $(F,G)\in\mathscr K_d$, with $d\ge8$, and set $\eta=uF+G$. Then the following are true:

\begin{enumerate}
    \item There exists a derivative-only polynomial density $K$, homogeneous
of length three and KdV weight $d-3$, such that the transformed
variations
\[
\overline F'=\overline F+\{\overline K,\overline H\},
\qquad
\overline\eta'=\overline\eta+\{\overline K,\overline J\}
\]
satisfy
\[
\{\overline H,\overline\eta'\}
+\{\overline F',\overline J\}=0,
\qquad
\partial_u\overline\eta'=\overline F',
\]
have the same quadratic components as the original variations,
and obey $F'_3=0$.
\item If $d=2m+4$ is even, write the quadratic canonical densities as
$c u_m^2$ and $\beta u_{m+1}^2$. Then
\[
\beta=-\frac{4(m+2)}{2m+3}c,
\qquad
\eta'_3=\frac{(-1)^m c}{9}
\left(p_{2m}+\frac{p_{2m+3}}{(2m+3)e_3}\right).
\]
In particular, if $c=0$, both transformed variations have
no components of length below four.

\item If $d$ is odd, both transformed variations have
no components of length below four.
\end{enumerate}
\end{lemma}

\subsection{Injectivity of the quadratic coefficient}\label{sec:quadratic-injectivity}
We prove Steps 2 and 3 together by showing that a pair is determined by the quadratic coefficient of its first component.
\begin{proposition}\label{prop:quad}
A pair in $\mathscr K_d$ whose first component has zero quadratic part is zero. Thus $\dim\mathscr K_d\le1$ in even $d$, and $\mathscr K_d=0$ in odd $d$.
\end{proposition}
\begin{proof}
We apply the transformation of Lemma~\ref{lem:cubic-reduction}(1).
By parts~(2) and~(3), the transformed variations have no components
of length below four: in even weights this follows from the
hypothesis $c=0$, and in odd weights it holds automatically.
We suppress primes on the transformed symbols $F_l$ and $\eta_l$
throughout the following elimination.

For a momentum tuple $\boldsymbol a$, write $e_j(\boldsymbol a)$ and $p_j(\boldsymbol a)$ for its elementary symmetric polynomial and power sum. We display the tuple whenever different sets of momenta occur in the same formula. Bare $e_j,p_j$ refer to the single tuple specified in the surrounding calculation. Since both variations start in length four, only the quadratic background terms contribute at that length. In four variables $p_3=3e_3$ and $p_5=-5e_2e_3$, so \eqref{eq:airy} gives
\[
 -6e_3\eta_4-12e_2e_3F_4
 =-6e_3(\eta_4+2e_2F_4)=0.
\]
The polynomial ring $\C[e_2,e_3,e_4]$ has no zero divisors; hence
$\eta_4=-2e_2F_4$.

The first variation remains derivative-only under the transformations, so $e_4\mid F_4$. Separate the terms containing $e_3$ by writing
\[
 F_4=e_4B(e_2,e_4)+e_3e_4C(e_2,e_3,e_4).
\]
Choose a derivative-only density of length four with symbol
$K_4=-e_4C/6$. Its bracket with the quadratic part of $H$ changes $F_4$ by $2p_3K_4=6e_3K_4$, giving
\[
 F'_4=F_4+6e_3K_4=e_4B(e_2,e_4).
\]
The cubic part of $H$ contributes only in length five. The same transformation applied to the second variation preserves commutativity and specialness and introduces no terms of length below four. In particular, $\eta'_4=-2e_2F'_4.$
We again suppress primes, so from now on
\[
 F_4=e_4B(e_2,e_4),\qquad \eta_4=-2e_2F_4.
\]

At length five, the unknown symbols $F_5,\eta_5$ enter through the quadratic background terms. The other contributions are the brackets of the length-four variations with the cubic background densities $u^3/6$ and $-uu_1^2$. Write
\[
 \boldsymbol\xi=(\xi_1,\ldots,\xi_5),\qquad
 s_{ij}=\xi_i+\xi_j,\qquad
 \boldsymbol\zeta_{ij}=(\xi_{\widehat{ij}},s_{ij}),
\]
where $\xi_{\widehat{ij}}$ is the tuple obtained by deleting $\xi_i,\xi_j$. With this notation, we have e.g. $e_2(\boldsymbol\zeta_{ij})=e_2(\boldsymbol\xi)-\xi_i\xi_j$.
The quartic relation on a merged tuple reads
\[
 \eta_4(\boldsymbol\zeta_{ij})
 =-2e_2(\boldsymbol\zeta_{ij})F_4(\boldsymbol\zeta_{ij}).
\]
Substituting it into \eqref{eq:symbracket} gives the combined source symbol
\begin{equation}\label{eq:L5}
 \mathcal S_5(F_4)(\boldsymbol\xi)=\frac25\sum_{i<j}s_{ij}
 \bigl(e_2(\boldsymbol\xi)+\xi_i^2+\xi_j^2\bigr)
 F_4(\boldsymbol\zeta_{ij}).
\end{equation}

In the next equation and congruence, all bare $e_j,p_j$ are evaluated on the full five-tuple $\boldsymbol\xi$.
Using $p_5=5(e_5-e_2e_3)$, the full length-five equation is therefore
\[
 -6e_3\eta_5+12(e_5-e_2e_3)F_5+\mathcal S_5(F_4)=0.
\]
Since $F_5$ is derivative-only, write $F_5=e_5C_5$. The first two terms then belong to the ideal $(e_3,e_5^2)$, so
\[
 \mathcal S_5(F_4)\equiv0\pmod{(e_3,e_5^2)}.
\]
Equivalently, after setting $e_3=0$, both the constant and the linear coefficient in $e_5$ must vanish.

\noindent\textbf{\textit{Claim.}} The polynomial $B$ satisfies
\begin{equation}\label{eq:operator}
 \mathscr D B=0,\qquad
 \mathscr D=4e_4\partial_{e_2}
 +2e_2e_4\partial_{e_4}+3e_2,
\end{equation}
where $e_2,e_4$ denote the two formal arguments of $B$.

\begin{proof}
For this computation, use a curve of five-tuples $\boldsymbol\xi(t)$ with
\[
 \boldsymbol\xi(0)=(r,-r,s,-s,0),\qquad A=r^2,\quad B_0=s^2,
\]
and velocity
\[
 \dot\xi_5(0)=1,\qquad
 \dot\xi_1(0)=\dot\xi_2(0)=\frac{B_0}{2(A-B_0)},\qquad
 \dot\xi_3(0)=\dot\xi_4(0)=-\frac{A}{2(A-B_0)}.
\]
Set $e_j^0:=e_j(\boldsymbol\xi(0))$. These are fixed values, namely
\[
 e_2^0=-A-B_0,\qquad e_3^0=e_5^0=0,\qquad e_4^0=AB_0.
\]
Direct differentiation gives
\[
 \left.\frac{d}{dt}e_j(\boldsymbol\xi(t))\right|_{t=0}=0
 \quad(1\le j\le4),\qquad
 \left.\frac{d}{dt}e_5(\boldsymbol\xi(t))\right|_{t=0}=e_4^0.
\]
Thus differentiating a symmetric polynomial along this curve extracts $e_4^0$ times its coefficient of $e_5$ at $e_3=e_5=0$, evaluated at $(e_2,e_4)=(e_2^0,e_4^0)$.

We spell out this derivative for \eqref{eq:L5}. For the four pairs $(i,5)$, define the elementary symmetric polynomials on the moving merged tuples by
\[
 E_{2,i}(t):=e_2(\boldsymbol\zeta_{i5}(t)),\qquad
 E_{4,i}(t):=e_4(\boldsymbol\zeta_{i5}(t)),\qquad 1\le i\le4.
\]
Their values and derivatives at $t=0$ are
\[
 E_{2,i}(0)=e_2^0,\quad E_{4,i}(0)=e_4^0,\qquad
 \dot E_{2,i}(0)=-\xi_i(0),\quad
 \dot E_{4,i}(0)=\frac{e_4^0}{\xi_i(0)}.
\]
The equalities of values hold at the base point only; the merged and full elementary symmetric polynomials have different derivatives. Write
\[
 T_i(t)=(\xi_i(t)+\xi_5(t))
 \bigl(e_2(\boldsymbol\xi(t))+\xi_i(t)^2+\xi_5(t)^2\bigr).
\]
At $t=0$, direct calculation gives
\[
 \sum_{i=1}^4 T_i(0)\dot E_{2,i}(0)=4e_4^0,\qquad
 \sum_{i=1}^4 T_i(0)\dot E_{4,i}(0)=2e_2^0e_4^0,\qquad
 \sum_{i=1}^4\left.\frac{d}{dt}(T_iE_{4,i})\right|_{t=0}
 =3e_2^0e_4^0.
\]
Let us now study
\[ \left.\frac{d}{dt}\mathcal S_5(F_4)(\boldsymbol\xi(t))\right|_{t=0} = \left. \frac25 \frac{d}{dt} \sum_{i<j}s_{ij}(t)
 \bigl(e_2(\boldsymbol\xi(t))+\xi_i(t)^2+\xi_j(t)^2\bigr)
 F_4(\boldsymbol\zeta_{ij}(t))\right|_{t=0}.
\]
The four cross pairs joining one of $r,-r$ to one of $s,-s$ contribute zero to the derivative. Since both
$e_2(\boldsymbol\xi(0))+\xi_i(0)^2+\xi_j(0)^2$
and $e_4(\boldsymbol\zeta_{ij}(0))$ vanish, while all
remaining factors in the summand of \eqref{eq:L5} are regular
at $t=0$, the summand is $O(t^2)$ and hence has zero first
derivative at $t=0$.

The two opposite pairs $(1,2)$ and $(3,4)$ also contribute zero. For the first of these,
\[
 e_4(\boldsymbol\zeta_{12}(t))
 =\xi_3(t)\xi_4(t)\xi_5(t)
   \bigl(\xi_1(t)+\xi_2(t)\bigr)=O(t^2),
\]
since $\xi_5(0)=0$ and $\xi_1(0)+\xi_2(0)=r-r=0$. Thus this summand again has zero first derivative; the pair $(3,4)$ is treated identically.

Only the four pairs $(i,5)$ remain. Put $a=e_2^0$ and $b=e_4^0$ for the following calculation, and write $\partial_1B,\partial_2B$ for the derivatives of $B$ with respect to its first and second arguments. Since $E_{2,i}(0)=a$ and $E_{4,i}(0)=b$, the product and chain rules give
\[
\begin{aligned}
 &\left.\frac{d}{dt}
 \left[T_i(t)E_{4,i}(t)
 B\bigl(E_{2,i}(t),E_{4,i}(t)\bigr)\right]\right|_{t=0}\\
 &\quad=\left.\frac{d}{dt}(T_iE_{4,i})\right|_{t=0}B(a,b)
 +bT_i(0)\dot E_{2,i}(0)\,\partial_1B(a,b)\\
 &\qquad\quad
 +bT_i(0)\dot E_{4,i}(0)\,\partial_2B(a,b).
\end{aligned}
\]
Summing over $i=1,\ldots,4$ and using the three identities above yields
\[
\begin{aligned}
 \left.\frac{d}{dt}\mathcal S_5(F_4)(\boldsymbol\xi(t))\right|_{t=0}
 &=\frac25\left(3abB(a,b)
       +4b^2\partial_1B(a,b)+2ab^2\partial_2B(a,b)\right)\\
 &=\frac25b\,(\mathscr DB)(a,b).
\end{aligned}
\]
Thus the two derivative terms in $\mathscr D$ come from differentiating the two merged arguments of $B$, while the multiplication term $3e_2B$ comes from differentiating the prefactor $T_iE_{4,i}$.

It remains to use the congruence modulo $(e_3,e_5^2)$. It says that there are polynomials $U,V$ in the full five-tuple elementary symmetric coordinates such that
\[
 \mathcal S_5(F_4)(\boldsymbol\xi)
 =e_3(\boldsymbol\xi)U(\boldsymbol\xi)
  +e_5(\boldsymbol\xi)^2V(\boldsymbol\xi).
\]
Here $U(\boldsymbol\xi),V(\boldsymbol\xi)$ denote their evaluations on those coordinates. Recall that along the chosen curve,
\[
 e_3(\boldsymbol\xi(0))=0,\qquad
 \left.\frac{d}{dt}e_3(\boldsymbol\xi(t))\right|_{t=0}=0,
 \qquad e_5(\boldsymbol\xi(0))=0.
\]
Comparing with the computed derivative gives
\[
 b\,(\mathscr DB)(a,b)=0.
\]
For $r,s$ with $rs\ne0$ and $r^2\ne s^2$, we have $b\ne0$ and the chosen velocities are defined, so $(\mathscr DB)(a,b)=0$. These parameters yield all pairs $(a,b)$ with $b\ne0$ and $a^2-4b\ne0$: choose $r^2,s^2$ as the two roots of $z^2+az+b$. This is a dense open subset of $\C^2$. Since $\mathscr DB$ is a polynomial in its two arguments, it vanishes identically, proving \eqref{eq:operator}.
\end{proof}

The operator $\mathscr D$ is injective on $\C[e_2,e_4]$. Indeed, if $B\ne0$, write
\[
 B=e_4^b b_0(e_2)+\text{terms divisible by }e_4^{b+1},
 \qquad b_0\ne0,
\]
where $b\ge0$ is minimal. The term $4e_4\partial_{e_2}$ raises the exponent of $e_4$, while the other two terms preserve it. Consequently,
\[
 [e_4^b]\mathscr D B=(2b+3)e_2b_0(e_2)\ne0,
\]
contradicting \eqref{eq:operator}. Thus $B=0$, and both $F_4$ and $\eta_4=-2e_2F_4$ vanish. Both variations now start in length at least five.

Suppose the first remaining length is $l\ge5$. In this step, $e_j,p_j$ are evaluated on $\boldsymbol\xi=(\xi_1,\ldots,\xi_l)$. Since every lower component is zero, the leading pair equation has only the quadratic-background contributions:
\[
 -6e_3\eta_l+12(e_5-e_2e_3)F_l=0.
\]
Reducing modulo $e_3$ gives $e_5F_l\equiv0\pmod{e_3}$. In the polynomial ring $\C[e_2,\ldots,e_l]$, $e_3$ and $e_5$ are coprime, hence $e_3\mid F_l$. Since $F_l$ is derivative-only, we also have $e_l\mid F_l$; since $l\ge5$, $e_l$ and $e_3$ are coprime. Therefore
\[
 K_l=-\frac{F_l}{6e_3}
\]
is polynomial and divisible by $e_l$, and so represents a derivative-only density of length $l$. Apply its transformation to both variations. At the leading length,
\[
 F'_l=F_l+6e_3K_l=0,\qquad
 \eta'_l=\eta_l+\frac{12}{5}p_5K_l.
\]
The transformed pair equation then reads $-6e_3\eta'_l=0$, so $\eta'_l=0$. Brackets with the higher-length background terms contribute only above length $l$, and no previously eliminated component is reintroduced.

Repeat this step. Every generator has KdV weight $d-3$, because the Airy bracket raises weight by three. The variations keep weights $d$ and $d+2$, and every jet factor has weight at least two; hence their monomial lengths are bounded. The elimination therefore terminates after finitely many transformations.

Finally, return to the original pair. Each transformation adds brackets with the fixed background Hamiltonians $H,J$, so the successive generators add. If their sum is $K_{\mathrm{tot}}$, the final first variation satisfies
\[
 0=\overline F+\{\overline K_{\mathrm{tot}},\overline H\}.
\]
Choose the derivative-only canonical representative $K$ of $-\overline K_{\mathrm{tot}}$. Then the original first variation obeys
\[
 \overline F=\{\overline K,\overline H\}.
\]
Its canonical density $F\in V_d$ contains no $u_1$. Lemma~\ref{lem:topv} therefore forces $K=0$ and $F=0$. The original pair equation now reduces to $\{\overline H,\overline G\}=0$; since $G$ is derivative-only, the centralizer assertion of the same lemma gives $G=0$.

This proves injectivity of the quadratic coefficient. In even weight there is only one possible quadratic canonical monomial, while in odd weight there is none, giving the stated dimension conclusions.
\end{proof}

\subsection{The quartic obstruction}\label{sec:resonance}
We prove Step 4 by deriving a necessary equation for the quartic component and showing that its forcing has a nonzero obstruction.

\begin{proposition}\label{prop:quartic-vanishing}
For every integer $d\ge8$ with $d\equiv2\pmod4$, the pair kernel $\mathscr K_d$ is zero.
\end{proposition}
\begin{proof}
Write $d=2m+4$, so $m\ge3$ is odd. Suppose that $\mathscr K_d$ contains a nonzero pair. By Proposition~\ref{prop:quad}, its quadratic coefficient is nonzero. Normalize that coefficient to $c=1$ and apply Lemma~\ref{lem:cubic-reduction}(1)--(2). The transformed pair has $F_3=0$ and cubic second-component symbol $U_3$ given by \eqref{eq:Usimple}, the $c=1$ case of the lemma. We suppress primes on the transformed symbols throughout this proof. We will derive \eqref{eq:forced} and then apply a linear functional that annihilates its left-hand side but is nonzero on its right-hand side.

\subsubsection{The quartic inhomogeneous term}
Here $\boldsymbol\xi=(\xi_1,\ldots,\xi_4)$ is a four-tuple with sum zero. For $\{k,l\}=\{1,2,3,4\}\setminus\{i,j\}$, put $s_{ij}=\xi_i+\xi_j$ and $\boldsymbol\zeta_{ij}=(\xi_k,\xi_l,s_{ij})$, a three-tuple. With $e_j=e_j(\boldsymbol\xi)$, the length-four component of the pair equation is
\[
 -6e_3\eta_4-12e_2e_3F_4
 -\frac14\sum_{i<j}s_{ij}U_3(\boldsymbol\zeta_{ij})
 -\frac{(-1)^m}{12}p_{2m+1}(\boldsymbol\xi)=0.
\]
The first two terms come from the quadratic backgrounds. By \eqref{eq:symbracket}, the sum comes from pairing $u^3/6$ with the cubic second variation, while the last term comes from pairing the quadratic first variation with $u^4/24$. The bracket of $F_3$ with $-uu_1^2$ vanishes because $F_3=0$. Solving for $\eta_4$ gives
\begin{equation}\label{eq:eta4}
 \eta_4=-2e_2F_4+J_m,
\end{equation}
where
\begin{equation}\label{eq:Jraw}
 J_m=\frac1{6e_3(\boldsymbol\xi)}\left(
 -\frac14\sum_{i<j}s_{ij}U_3(\boldsymbol\zeta_{ij})
 -\frac{(-1)^m}{12}p_{2m+1}(\boldsymbol\xi)\right).
\end{equation}
Every $e_j,p_j$ inside $U_3(\boldsymbol\zeta_{ij})$ is evaluated on that three-tuple, as in \eqref{eq:Usimple}. Let
\[
 \boldsymbol\rho=(\xi_1+\xi_2,\xi_1+\xi_3,\xi_1+\xi_4),\qquad
 S_n(\boldsymbol\xi):=p_n(\boldsymbol\rho).
\]
For even $n$, this is the symmetric polynomial obtained by taking one pair sum from each complementary pair of pairs. Then
\begin{equation}\label{eq:J}
 J_m=\frac{(-1)^{m+1}}{216(2m+3)}
 \left((2m+2)\frac{p_{2m+1}(\boldsymbol\xi)}{e_3(\boldsymbol\xi)}
       +\frac{p_{2m+2}(\boldsymbol\xi)-S_{2m+2}(\boldsymbol\xi)}{e_4(\boldsymbol\xi)}\right).
\end{equation}
This is a polynomial divisible by $e_4$, and it is even in $e_3$.

The following two merging identities prove \eqref{eq:J} directly.
\begin{align}
 \sum_{i<j}s_{ij}p_{2m}(\boldsymbol\zeta_{ij})
 &=-2p_{2m+1}(\boldsymbol\xi),\\
 \sum_{i<j}s_{ij}\frac{p_n(\boldsymbol\zeta_{ij})}{e_3(\boldsymbol\zeta_{ij})}
 &=\frac{e_3(\boldsymbol\xi)}{e_4(\boldsymbol\xi)}
 \bigl(p_{n-1}(\boldsymbol\xi)-S_{n-1}(\boldsymbol\xi)\bigr)
 -p_{n-2}(\boldsymbol\xi),\quad n\ge3\text{ odd}.
\end{align}
For the second identity, group the six summands into three complementary pairs. In one such pair write $a=\xi_i$, $b=\xi_j$, $c=\xi_k$, $d=\xi_l$ and $s=a+b=-(c+d)$. Since $n$ is odd,
\[
\begin{aligned}
 \frac{p_n(c,d,s)}{cd}+\frac{p_n(a,b,-s)}{ab}
 &=\frac{c^n+d^n}{cd}+\frac{a^n+b^n}{ab}
   +s^n\left(\frac1{cd}-\frac1{ab}\right)\\
 &=\frac{c^n+d^n}{cd}+\frac{a^n+b^n}{ab}
   -\frac{e_3(\boldsymbol\xi)}{e_4(\boldsymbol\xi)}s^{n-1},
\end{aligned}
\]
where we used
\[
 e_3(\boldsymbol\xi)
 =s_{ij}(\xi_k\xi_l-\xi_i\xi_j),\qquad
 \{k,l\}=\{1,2,3,4\}\setminus\{i,j\}.
\]
Summing the first two fractions over the three complementary pairs gives
\[
 \sum_{i\ne j}\frac{\xi_i^{n-1}}{\xi_j}
 =\frac{e_3(\boldsymbol\xi)}{e_4(\boldsymbol\xi)}
   p_{n-1}(\boldsymbol\xi)-p_{n-2}(\boldsymbol\xi).
\]
The remaining three powers of $s$ sum to $S_{n-1}(\boldsymbol\xi)$, proving the second merging identity. Substitution of both identities into \eqref{eq:Jraw}, using \eqref{eq:Usimple}, gives \eqref{eq:J}.

These formulas define an $e_4$-divisible polynomial $J_m$ for every integer $m\ge3$, independently of the existence of a kernel pair. Indeed, $U_3$ in \eqref{eq:Usimple} is polynomial, so the numerator of \eqref{eq:Jraw} is a symmetric polynomial in the four momenta. It is odd under simultaneous negation of the momenta; in $\C[e_2,e_3,e_4]$ this implies divisibility by $e_3$, and the quotient is even in $e_3$. To check divisibility by $e_4$, specialize to $(a,b,c,0)$ with $a+b+c=0$. The three pairs containing the zero momentum contribute $(a+b+c)U_3(a,b,c)=0$ to the sum in \eqref{eq:Jraw}. For the other three pairs, use
\[
 3U_3(q,-q,0)=(-1)^m q^{2m}
\]
to obtain $-(-1)^m p_{2m+1}(a,b,c)/3$. The two terms in the numerator of \eqref{eq:Jraw} therefore cancel. Hence $J_m$ vanishes on $e_4=0$ and is divisible by $e_4$.

\subsubsection{The forcing polynomial}
Apply the quartic canonical change used in Proposition~\ref{prop:quad}: write
$F_4=e_4B(e_2,e_4)+e_3e_4Q(e_2,e_3,e_4)$ and choose the derivative-only generator with symbol $K_4=-e_4Q/6$. It changes the two components by
\[
 \Delta F_4=6e_3K_4,\qquad
 \Delta\eta_4=-12e_2e_3K_4.
\]
Thus the transformed first component is $F_4=e_4B(e_2,e_4)$, while $\eta_4+2e_2F_4=J_m$ is unchanged. We again suppress primes. From now on, write $J_m(e_2,e_3,e_4)$ for its expression in elementary symmetric coordinates; these three arguments are scalar coordinate values, not a three-tuple of momenta. Set
\[
 R_m(e_2,e_4)=J_m(e_2,0,e_4),\qquad H_m(e_2,e_3)=[e_4]J_m(e_2,e_3,e_4).
\]
For auxiliary $r,s$ with $e_2=-(r^2+s^2)$ and $e_4=r^2s^2$, define
\begin{align}\label{eq:C}
 C_m(e_2,e_4)=2rs\bigl(& (r+s)^2H_m(e_2-rs,rs(r+s))\nonumber\\
                 &-(r-s)^2H_m(e_2+rs,-rs(r-s))\bigr).
\end{align}
The arguments of $H_m$ in \eqref{eq:C} are likewise elementary symmetric coordinate values. The expression is symmetric in $r^2,s^2$, so it is a polynomial in $e_2,e_4$. The length-five compatibility equation modulo $(e_3,e_5^2)$ is
\begin{equation}\label{eq:forced}
 \boxed{\mathscr D B=\Omega_m(e_2,e_4),\qquad
 \Omega_m=\frac{(2m+1)R_m+C_m}{2e_4}.}
\end{equation}
Note that $rs\mid C_m$ and $C_m$ is even in each of $r,s$, so $e_4=r^2s^2\mid C_m$. Also, $e_4\mid R_m$ because $e_4\mid J_m$. Thus the quotient defining $\Omega_m$ is polynomial.

We use ordinary momentum degree, with $\deg e_j=j$. By \eqref{eq:J}, $R_m$ has degree $2m-2$, while $H_m=[e_4]J_m$ has degree $2m-6$. The prefactors $rs(r\pm s)^2$ in \eqref{eq:C} restore four degrees, so $C_m$ also has degree $2m-2$. Dividing by $e_4$ gives $\deg\Omega_m=2m-6$. Moreover, $F_4$ has momentum degree $d-8=2m-4$, so $F_4=e_4B$ implies $\deg B=2m-8$. In particular, $B=0$ when $m=3$.

To prove \eqref{eq:forced}, return to five momenta $\boldsymbol\xi=(\xi_1,\ldots,\xi_5)$ and the merged four-tuples $\boldsymbol\zeta_{ij}=(\xi_{\widehat{ij}},s_{ij})$, where $s_{ij}=\xi_i+\xi_j$. The $J_m$ contribution to the length-five source is
\[
 -\frac15\sum_{i<j}s_{ij}
 J_m\bigl(e_2(\boldsymbol\zeta_{ij}),e_3(\boldsymbol\zeta_{ij}),
          e_4(\boldsymbol\zeta_{ij})\bigr).
\]
We repeat the tangent calculation in Proposition~\ref{prop:quad}, using the same curve $\boldsymbol\xi(t)$ through $(r,-r,s,-s,0)$ with the velocities specified there. Differentiating at $t=0$ extracts $e_4^0$ times the coefficient of the full-tuple coordinate $e_5$ after setting $e_3=0$, evaluated at $(e_2,e_4)=(e_2^0,e_4^0)$.

For the sum before multiplication by $-1/5$, the four pairs containing the zero momentum contribute
\[
 3R_m+2e_2\partial_{e_2}R_m+4e_4\partial_{e_4}R_m
 =(2m+1)R_m,
\]
evaluated at $(e_2^0,e_4^0)$. The first term differentiates the factors $s_{i5}$; the other two differentiate the merged coordinates, as in Proposition~\ref{prop:quad}. The equality is Euler's identity for the momentum degree $2m-2$ of $R_m$. There is no $e_3$-derivative contribution because $J_m$ is even in $e_3$. The four cross pairs contribute $C_m$ from \eqref{eq:C}: the merged $e_4$ vanishes at the base point, so only the coefficient $H_m=[e_4]J_m$ contributes to the derivative. The opposite pairs contribute zero, since their merged $e_4$ is $O(t^2)$ and $e_4\mid J_m$. Thus the derivative of the $J_m$ source is
\[
 -\frac15\bigl((2m+1)R_m+C_m\bigr)(e_2^0,e_4^0).
\]
The part of the source depending on $B$ has derivative $(2/5)e_4^0(\mathscr DB)(e_2^0,e_4^0)$, as computed in Proposition~\ref{prop:quad}. The full source belongs to $(e_3,e_5^2)$ in the full five-tuple coordinates, so its derivative must vanish. Consequently,
\[
 2e_4^0(\mathscr DB)(e_2^0,e_4^0)
 =\bigl((2m+1)R_m+C_m\bigr)(e_2^0,e_4^0).
\]
For generic $r,s$, we may divide by $2e_4^0$. Since the base-point coordinates range over a dense set and the quotient is polynomial, this proves \eqref{eq:forced} as a polynomial identity in $e_2,e_4$.

\subsubsection{The nonzero obstruction}\label{sec:obstruction}
We now show that the necessary equation \eqref{eq:forced} has no polynomial solution for odd $m$.
Define the linear functional $\mathcal A:\C[e_2,e_4]\to\C$ by
\begin{equation}\label{eq:moment}
 \mathcal A(T)=\int_{-1}^{1}\sqrt{1-z^2}\,
 T\left(z,\frac{z^2-1}{4}\right)\dd z.
\end{equation}

\noindent\textbf{\textit{Claim (Annihilation).}}
For every $B\in\C[e_2,e_4]$, we have $\mathcal A(\mathscr DB)=0$.
\begin{proof}[Proof of the claim]
Set $b(z)=B\bigl(z,(z^2-1)/4\bigr)$. By the chain rule,
\[
 b'(z)=\left[\frac{\partial B}{\partial e_2}
       +\frac z2\frac{\partial B}{\partial e_4}\right]_
       {(e_2,e_4)=(z,(z^2-1)/4)},
 \qquad
 (\mathscr DB)\left(z,\frac{z^2-1}{4}\right)
 =(z^2-1)b'(z)+3zb(z).
\]
Consequently,
\[
\begin{aligned}
 \mathcal A(\mathscr DB)
 &=\int_{-1}^{1}\sqrt{1-z^2}\,
       \bigl((z^2-1)b'(z)+3zb(z)\bigr)\dd z\\
 &=-\int_{-1}^{1}\frac{d}{dz}
       \bigl((1-z^2)^{3/2}b(z)\bigr)\dd z
 =-\left[(1-z^2)^{3/2}b(z)\right]_{-1}^{1}=0.
\end{aligned}
\]
\end{proof}

\noindent\textbf{\textit{Claim (Explicit obstruction).}}
For every odd $m\ge3$,
\begin{equation}\label{eq:obstruction}
 \boxed{\mathcal A(\Omega_m)
 =-\frac{\pi}{27}(m-1)(2m+1)
   \frac{\binom{2m}{m}}{4^m}<0.}
\end{equation}
\begin{proof}[Proof of the claim]
Introduce a generating variable $t$ and put $\theta=t\partial_t$. We use the explicit formulas \eqref{eq:J}, \eqref{eq:C}, and \eqref{eq:forced} as algebraic definitions for all integer indices $n\ge3$. The polynomiality argument above applies to every such index; this extension does not assume the existence of kernel pairs at those weights. Set
\[
 \kappa_n=216(2n+3)(-1)^{n+1},\qquad
 \begin{aligned}
 \mathcal R(e_2,e_4;t)&=\sum_{n\ge3}\kappa_nR_n(e_2,e_4)t^n,\\
 \mathcal H(e_2,e_3;t)&=\sum_{n\ge3}\kappa_nH_n(e_2,e_3)t^n.
 \end{aligned}
\]
The signs and denominators in \eqref{eq:J} are thereby removed.  Write
\[
 Q_0=1+e_2t+e_4t^2,\qquad
 S_0=1+2e_2t+(e_2^2-4e_4)t^2,
 \qquad a=1+e_2t,\quad Y=e_3^2t^3.
\]
We derive the two generating-function identities
\begin{align}\label{eq:forcing-newton}
 \mathcal R
 &=(4\theta^2+6\theta+2)\frac{t}{Q_0}
   +\frac1{e_4}\partial_t\log\frac{S_0}{Q_0^2},\nonumber\\
 \mathcal H
 &=(\theta+1)\frac{t^3(3Y-5a^2-8a-8)}{(a^2-Y)^2}.
\end{align}
Throughout this calculation, $e_j$ and $p_j$ refer to the four-tuple
$\boldsymbol\xi$ in \eqref{eq:J}, whereas
$S_j=p_j(\boldsymbol\rho)$ for
$\boldsymbol\rho=(\xi_1+\xi_2,\xi_1+\xi_3,\xi_1+\xi_4)$.
Multiplying \eqref{eq:J} by $\kappa_n$ gives
\[
 \kappa_nJ_n
 =(2n+2)\frac{p_{2n+1}}{e_3}
   +\frac{p_{2n+2}-S_{2n+2}}{e_4}.
\]
To sum the two terms, introduce
\[
 \mathcal P(t)=\sum_{n\ge1}\frac{p_{2n+1}}{(2n+1)e_3}t^{n-1},
 \qquad
 \mathcal L(t)=\sum_{k\ge1}\frac{p_{2k}-S_{2k}}{k}t^k.
\]
Since $\theta(t^n)=nt^n$, the odd-power contribution sums to
\[
 \sum_{n\ge1}(2n+2)\frac{p_{2n+1}}{e_3}t^n
 =(2\theta+2)(2\theta+1)\bigl(t\mathcal P(t)\bigr).
\]
Differentiating $\mathcal L$ and shifting $k=n+1$ gives
\[
 \sum_{n\ge1}(p_{2n+2}-S_{2n+2})t^n
 =\partial_t\mathcal L(t),
\]
because $p_2-S_2=0$. The right-hand side of the normalized formula
for $J_n$ vanishes for $n=1,2$: its two terms are respectively
$12,-12$ and $-30e_2,30e_2$. Thus the sum starting at $n=3$ is
\[
 \sum_{n\ge3}\kappa_nJ_nt^n
 =(2\theta+2)(2\theta+1)\bigl(t\mathcal P(t)\bigr)
   +\frac1{e_4}\partial_t\mathcal L(t).
\]

We now evaluate the odd-power series $\mathcal P$ and then the
even-power series $\mathcal L$. Newton's identities in logarithmic form read
\[
 \sum_{k\ge1}\frac{p_k}{k}v^k
 =-\log\prod_{i=1}^4(1-\xi_i v)
 =-\log(1+e_2v^2-e_3v^3+e_4v^4).
\]
Taking the odd part, dividing by $e_3v^3$, and writing $t=v^2$ yields
\[
 \mathcal P(t)
 =\sum_{j\ge0}\frac{Y^j}{(2j+1)Q_0^{2j+1}}.
\]
For the even powers, the product identities
\[
 \prod_{i=1}^4(1-\xi_i^2t)=Q_0^2-Y,
 \qquad
 \prod_{j=1}^3(1-\rho_j^2t)=S_0-Y
\]
likewise give
\[
 \mathcal L(t)=\log\frac{S_0-Y}{Q_0^2-Y}.
\]

To obtain $\mathcal R$, specialize the formula for
$\sum_{n\ge3}\kappa_nJ_nt^n$ at $e_3=0$. Since then $Y=0$, we have
\[
 \left.\mathcal P(t)\right|_{e_3=0}=\frac1{Q_0},
 \qquad
 \left.\mathcal L(t)\right|_{e_3=0}=\log\frac{S_0}{Q_0^2}.
\]
This proves the first identity in \eqref{eq:forcing-newton}, using
$(2\theta+2)(2\theta+1)=4\theta^2+6\theta+2$.

For the second identity, recall that $H_n=[e_4]J_n$ and put
$\Delta=a^2-Y$. As $Q_0=a+e_4t^2$, the series for $\mathcal P$ gives
\[
 [e_4]\mathcal P(t)
 =-t^2\sum_{j\ge0}\frac{Y^j}{a^{2j+2}}
 =-\frac{t^2}{\Delta}.
\]
The even-power contribution is divided by $e_4$, so we need
$[e_4^2]\mathcal L$. Expanding
\[
 \mathcal L(t)
 =\log(\Delta-4e_4t^2)
  -\log(\Delta+2ae_4t^2+e_4^2t^4)
\]
to second order in $e_4$ gives
\[
 [e_4^2]\mathcal L(t)
 =-\frac{8t^4}{\Delta^2}-\frac{t^4}{\Delta}
   +\frac{2a^2t^4}{\Delta^2}
 =\frac{t^4(a^2+Y-8)}{\Delta^2}.
\]
Consequently,
\begin{align*}
 \mathcal H
 & =-(2\theta+2)(2\theta+1)\frac{t^3}{\Delta}
    +\partial_t\frac{t^4(a^2+Y-8)}{\Delta^2}\\
 & =(\theta+1)\left[
    -2(2\theta+1)\frac{t^3}{\Delta}
    +\frac{t^3(a^2+Y-8)}{\Delta^2}\right],
\end{align*}
where we used $\partial_t(t^4G)=(\theta+1)(t^3G)$.
Finally, $\theta a=a-1$ and $\theta Y=3Y$, so
\[
 (2\theta+1)\frac{t^3}{\Delta}
 =\frac{t^3\bigl(7\Delta-2\theta\Delta\bigr)}{\Delta^2}
 =\frac{t^3(3a^2+4a-Y)}{\Delta^2}.
\]
Substitution gives the numerator $3Y-5a^2-8a-8$ in the second
identity of \eqref{eq:forcing-newton}.

Now restrict to the curve used in \eqref{eq:moment}, and put
\[
 Q(z,t)=1+zt+\frac{z^2-1}{4}t^2,\qquad
 S(z,t)=1+2zt+t^2,
\]
\[
 \mathcal O(z,t)=\sum_{n\ge3}\kappa_n
 \Omega_n\left(z,\frac{z^2-1}{4}\right)t^n.
\]
Define
\begin{equation}\label{eq:forcing-bases}
 f_1=\frac{2t}{(z^2-1)Q},\qquad
 f_2=-\frac{4t(t^2+3tz+6)}{(z^2-1)SQ},\qquad
 f_3=-\frac{4t^3(t^2+6tz+21)}{S^2Q}.
\end{equation}
Multiplying \eqref{eq:forced} by $\kappa_m t^m$, summing over
$m\ge3$, and restricting to
$(e_2,e_4)=(z,(z^2-1)/4)$ gives
\begin{equation}\label{eq:forcing-decomposition}
 \mathcal O
 =2(2\theta+1)^2(\theta+1)f_1
 +(2\theta+1)f_2+(\theta+1)f_3.
\end{equation}
Indeed, the contribution of $(2n+1)R_n/(2e_4)$ follows from the first identity in \eqref{eq:forcing-newton}, using
\[
 4\theta^2+6\theta+2=2(2\theta+1)(\theta+1),\qquad
 \left.\frac1{2e_4^2}\partial_t\log\frac{S_0}{Q_0^2}\right|_
 {(e_2,e_4)=(z,(z^2-1)/4)}=f_2.
\]
For the contribution of $C_n/(2e_4)$, put $q^2=(z^2-1)/4$. The two substitutions into $\mathcal H$ prescribed by \eqref{eq:C} have
\[
 a_\pm=1+(z\mp q)t,\qquad
 Y_\pm=q^2(-z\pm2q)t^3,
 \qquad
 a_\pm^2-Y_\pm=Q\bigl(1+(z\mp2q)t\bigr).
\]
Since the product of the last two linear factors is $S$, combining the two terms gives
\[
 \frac{t^3}{q}\left[
 \frac{(-z+2q)(3Y_+-5a_+^2-8a_+-8)}{(a_+^2-Y_+)^2}
 -\frac{(-z-2q)(3Y_--5a_-^2-8a_--8)}{(a_-^2-Y_-)^2}
 \right]=f_3.
\]
The factor $\theta+1$ in $\mathcal H$ then gives the last term of \eqref{eq:forcing-decomposition}.

It remains to integrate these three simpler rational functions. Set
\[
 I_j(t)=\int_{-1}^1\sqrt{1-z^2}\,f_j(z,t)\dd z,
 \qquad E(t)=(1-t)^{-1/2}-(1+t)^{-1/2}.
\]
For $t$ in a sufficiently small complex disk, $Q(z,t)$ and $S(z,t)$ are uniformly nonzero on $-1\le z\le1$. The weighted integrands, and their $t$-derivatives on smaller disks, are bounded by a constant multiple of $(1-z^2)^{-1/2}$. Thus these integrals are analytic near $t=0$, and differentiation and coefficient extraction may be performed under the integral.
Direct integration gives
\begin{align}\label{eq:forcing-base-integrals}
 I_1&=-2\pi E,\nonumber\\
 I_2&=\pi\left(-8\theta E+\frac{32t}{(1-t^2)^2}\right),\\
 I_3&=\pi\left(16(2\theta+1)E
              -\frac{16t(t^2+3)}{(1-t^2)^2}\right).\nonumber
\end{align}
Integrate \eqref{eq:forcing-decomposition} and substitute \eqref{eq:forcing-base-integrals}. The rational terms cancel because
\[
 (2\theta+1)\frac{32t}{(1-t^2)^2}
 -(\theta+1)\frac{16t(t^2+3)}{(1-t^2)^2}=0.
\]
Combining the remaining operators yields
\begin{equation}\label{eq:forcing-integrated}
 \sum_{n\ge3}\kappa_n\mathcal A(\Omega_n)t^n
 =-4\pi(2\theta+1)(2\theta+3)(\theta-1)E(t).
\end{equation}
Finally,
\[
 E(t)=2\sum_{\substack{n\ge1\\n\text{ odd}}}
       \frac{\binom{2n}{n}}{4^n}t^n.
\]
For odd $m\ge3$, extracting the coefficient of $t^m$ in \eqref{eq:forcing-integrated} gives
\[
 216(2m+3)\mathcal A(\Omega_m)
 =-8\pi(2m+1)(2m+3)(m-1)
       \frac{\binom{2m}{m}}{4^m}.
\]
Thus \eqref{eq:obstruction} follows; all remaining factors after the minus sign are positive.
\end{proof}

We can now finish the proof of the proposition. Applying $\mathcal A$ to~\eqref{eq:forced} yields a contradiction by the above claims: therefore no pair can have a nonzero quadratic coefficient, and Proposition~\ref{prop:quad} implies $\mathscr K_d=0$.
\end{proof}

\subsection{Diagonal injectivity}\label{sec:diagonal}
Only Step 5 remains. Proposition~\ref{prop:quad} gives $\dim\mathscr K_{4g}\le1$. We construct a pair using the polynomial KdV hierarchy and compute its nonzero quadratic and diagonal coefficients. The calculation gives the following more precise statement.

\begin{proposition}\label{prop:diagonal-algebraic}
For every $g\ge2$, the space $\mathscr K_{4g}$ contains a pair whose first component has coefficients
\begin{equation}\label{eq:algebraic-cA}
 c_g=[u_{2g-2}^2]F_g
 =\frac{(2g-2)6^{2g-2}}{(4g-3)!!},\qquad
 A_g=[u_2^g]F_g
 =\frac{(3g-2)(2g-1)}2\,
 \frac{24^g|B_{2g}|}{(2g)!}.
\end{equation}
In particular, both coefficients are nonzero. Every pair in $\mathscr K_{4g}$ satisfies
\begin{equation}\label{eq:diagonal-ratio}
 [u_2^g]F=\rho_g[u_{2g-2}^2]F,\qquad
 \rho_g=
 \frac{(3g-2)(4g-3)}{2g(g-1)6^{g-2}}
 \binom{4g-4}{2g-2}|B_{2g}|.
\end{equation}
Consequently the diagonal coefficient is injective on $\mathscr K_{4g}$.
\end{proposition}
Here the Bernoulli numbers are defined by $t/(e^t-1)=\sum_{n\ge0}B_n t^n/n!$.

\begin{proof}[Proof of Proposition~\ref{prop:diagonal-algebraic}]
The proof has four parts: construct a polynomial pair, compute its quadratic coefficient, express its standard diagonal through a linear functional, and evaluate that functional.

\subsubsection{A local polynomial pair from KdV conservation laws}\label{sec:local-pair}
We construct densities $F_g^0,\eta_g^0$ satisfying the pair equation and specialness, and then standardize them to obtain $(F_g,G_g)\in\mathscr K_{4g}$. We also show that this standardization preserves the quadratic coefficient, which will be computed in the next subsection.

For a local functional $\overline A$, put
$Q_A=\partial_x(\delta\overline A/\delta u)$. To calculate density
representatives of Poisson brackets, write
\[
 \mathcal E_Q(f)=\sum_{j\ge0}(\partial_x^jQ)\frac{\partial f}{\partial u_j},
 \qquad
 \{\overline f,\overline A\}
 =\int\mathcal E_{Q_A}(f)\dd x.
\]
The second identity follows by integration by parts. In particular,
\[
 Q_H=uu_1+2u_3,\qquad
 Q_J=\tfrac12u^2u_1+4u_1u_2+2uu_3+\tfrac{12}{5}u_5.
\]
Denote by $h_i^{\mathrm{KdV}}$, $i\ge-2$, the background KdV
Hamiltonian densities specified by the normalized diagonal formal heat kernel
\begin{equation}\label{eq:heat-U}
 U(T,x)=\sqrt{48\pi T}\,
 e^{T(12\Dx^2+u)}(x,x)
 =\sum_{n\ge0}h_{n-2}^{\mathrm{KdV}}T^n.
\end{equation}
We use the heat-kernel description of the KdV densities \cite[Corollary 3.3]{Iliev}, interpreted as a local formal expansion. Relative to the operator $\partial_x^2+u$ in that reference, our normalization replaces its potential by $u/12$ and its heat time by $12T$; its $n$th diagonal coefficient is therefore multiplied by $12^n$. Equivalently, these densities are the differential polynomials determined by the normalized Lenard--Magri recursion \cite{Magri}
\begin{equation}\label{eq:heat-recursion}
 h_{-2}^{\mathrm{KdV}}=1,\qquad
 \Dx h_{i+1}^{\mathrm{KdV}}=
 \frac{(2u\Dx+u_1+6\Dx^3)h_i^{\mathrm{KdV}}}{2i+5},
 \qquad i\ge-2,
\end{equation}
with the homogeneous integration constants equal to zero. For example,
\[
 h_{-1}^{\mathrm{KdV}}=u,\quad
 h_0^{\mathrm{KdV}}=\frac{u^2}{2}+2u_2,\quad
 h_1^{\mathrm{KdV}}=\frac{u^3}{6}+2uu_2+u_1^2+\frac{12}{5}u_4.
\]
Two generating-series identities give the lowering relations we need. A constant shift $u\mapsto u+c$ multiplies $U(T,x)$ by $e^{cT}$, so $\partial_uU=TU$. For the variational identity, put $\mathcal L=12\partial_x^2+u$. Duhamel's formula and cyclicity of the trace give
\[
\begin{aligned}
 \delta\operatorname{Tr}e^{T\mathcal L}
 &=\int_0^T\operatorname{Tr}
   \bigl(e^{(T-s)\mathcal L}\,\delta u\,e^{s\mathcal L}\bigr)\dd s\\
 &=T\operatorname{Tr}\bigl(\delta u\,e^{T\mathcal L}\bigr).
\end{aligned}
\]
Here the trace is interpreted coefficientwise in the integrated local heat expansion. Multiplying by the normalization in \eqref{eq:heat-U} yields the trace-variation identity
\[
 \delta\int U(T,x)\dd x
 =T\int U(T,x)\,\delta u(x)\dd x,
 \qquad
 \frac{\delta}{\delta u}\int U(T,x)\dd x=TU(T,x).
\]
Comparing coefficients in these two identities gives
\[
 \partial_u h_i^{\mathrm{KdV}}=h_{i-1}^{\mathrm{KdV}},\qquad
 \frac{\delta}{\delta u}\int h_i^{\mathrm{KdV}}\dd x
 =h_{i-1}^{\mathrm{KdV}},\qquad i\ge-1.
\]
These densities represent the usual normalized KdV Hamiltonians, with
$\overline{h_1^{\mathrm{KdV}}}=\overline H$ and
$\overline{h_2^{\mathrm{KdV}}}=\overline J$.
The Lenard--Magri recursion \eqref{eq:heat-recursion} proves their Poisson commutativity:
\[
 \{\overline{h_i^{\mathrm{KdV}}},\overline{h_j^{\mathrm{KdV}}}\}=0.
\]
Notice that the representatives $h_i^{\mathrm{KdV}}$ need not be canonical; this choice will simplify their generating series.

In particular,
$\{\overline{h_i^{\mathrm{KdV}}},\overline H\}
=\{\overline{h_i^{\mathrm{KdV}}},\overline J\}=0$.
Thus the bracket densities
$\mathcal E_{Q_H}(h_i^{\mathrm{KdV}})$ and
$\mathcal E_{Q_J}(h_i^{\mathrm{KdV}})$ are total derivatives.
Choose homogeneous polynomial primitives $P_i,S_i$, called fluxes,
satisfying
\[
 \Dx P_i=\mathcal E_{Q_H}(h_i^{\mathrm{KdV}}),\qquad
 \Dx S_i=\mathcal E_{Q_J}(h_i^{\mathrm{KdV}}),\qquad
 P_{-2}=S_{-2}=0.
\]
Set $h_{-3}^{\mathrm{KdV}}=P_{-3}=S_{-3}=0$. Differentiating these identities with respect to $u$ gives
\begin{equation}\label{eq:flux-shifts}
 \partial_uP_i=P_{i-1}+h_i^{\mathrm{KdV}}-\delta_{i,-2},\qquad
 \partial_uS_i=S_{i-1}+P_i.
\end{equation}
Indeed, applying $\Dx$ to each proposed equality proves it because
$\partial_uQ_H=u_1$ and $\partial_uQ_J=Q_H$. Homogeneity removes the resulting constants, except at $i=-2$, where the displayed delta term is checked directly.

For $n=2g-5$, define the finite local sums
\begin{equation}\label{eq:local-pair-construction}
 F_g^0=\sum_{i=-2}^{n+2}(-1)^i h_i^{\mathrm{KdV}}P_{n-i}-h_{n+3}^{\mathrm{KdV}},\qquad
 \eta_g^0=\sum_{i=-2}^{n+2}(-1)^i h_i^{\mathrm{KdV}}S_{n-i}-h_{n+4}^{\mathrm{KdV}}.
\end{equation}
They have KdV weights $4g$ and $4g+2$. Since $n$ is odd, differentiating the sums and using \eqref{eq:flux-shifts} makes the adjacent terms cancel, leaving
\begin{equation}\label{eq:local-pair-specialness}
 \partial_uF_g^0=0,\qquad \partial_u\eta_g^0=F_g^0.
\end{equation}
The boundary contribution $h_{n+2}^{\mathrm{KdV}}$ in the first calculation comes from the delta term in \eqref{eq:flux-shifts} and cancels $\partial_uh_{n+3}^{\mathrm{KdV}}$.

These densities also satisfy the pair equation. To verify it, first use
$\{\overline H,\overline J\}=0$ to obtain the compatibility of the fluxes.
Variational differentiation of the Poisson bracket gives
\[
 \mathcal E_{Q_J}(Q_H)-\mathcal E_{Q_H}(Q_J)
 =\partial_x\frac{\delta\{\overline H,\overline J\}}{\delta u}=0.
\]
The chain rule and the defining equations for $P_i,S_i$ therefore give
\begin{align*}
 \partial_x\bigl(\mathcal E_{Q_J}(P_i)-\mathcal E_{Q_H}(S_i)\bigr)
 &=\mathcal E_{Q_J}\mathcal E_{Q_H}(h_i^{\mathrm{KdV}})
   -\mathcal E_{Q_H}\mathcal E_{Q_J}(h_i^{\mathrm{KdV}})\\
 &=\sum_{r\ge0}\frac{\partial h_i^{\mathrm{KdV}}}{\partial u_r}
   \partial_x^r\bigl(\mathcal E_{Q_J}(Q_H)
                    -\mathcal E_{Q_H}(Q_J)\bigr)=0.
\end{align*}
Homogeneity removes the possible constant, hence
$\mathcal E_{Q_J}(P_i)=\mathcal E_{Q_H}(S_i)$.
Using this identity and the product rule to calculate the following
Poisson brackets yields
\begin{align*}
 &\left\{\overline{h_i^{\mathrm{KdV}}P_j-P_ih_j^{\mathrm{KdV}}},
          \overline J\right\}
 -\left\{\overline{h_i^{\mathrm{KdV}}S_j-S_ih_j^{\mathrm{KdV}}},
          \overline H\right\}\\
 &\hspace{25mm}=\int\partial_x(S_iP_j-P_iS_j)\dd x=0.
\end{align*}
For $i+j=n$, oddness of $n$ gives
\[
 \sum_{i+j=n}(-1)^i h_i^{\mathrm{KdV}}P_j
 =\frac12\sum_{i+j=n}(-1)^i
   (h_i^{\mathrm{KdV}}P_j-P_ih_j^{\mathrm{KdV}}),
\]
where both indices are at least $-2$; the same identity holds with $S$
in place of $P$. Sum the preceding bracket identity with these signs.
The subtracted KdV Hamiltonians in \eqref{eq:local-pair-construction}
Poisson-commute with $\overline H,\overline J$, so
\[
 \{\overline H,\overline{\eta_g^0}\}
 +\{\overline{F_g^0},\overline J\}=0.
\]

Canonically reduce the pair. Mimicking the elimination procedure in the proof of \cite[Theorem 3.8]{BR}, we use Lemma~\ref{lem:topv} to standardize the infinitesimal pair in finitely many steps. To cancel a term $c\,u_1^sM$ at the largest $u_1$-exponent $s\ge1$, where $M$ contains only jets of index at least two and has length $\ell$, choose
\[
 K=\frac{c}{4g-2s-\ell-2}\,u_1^{s-1}M
\]
and apply the local canonical change
\[
 \overline F'=\overline F-\{\overline K,\overline H\},\qquad
 \overline\eta'=\overline\eta-\{\overline K,\overline J\}.
\]
The denominator is the positive coefficient in that lemma: the canonical
representative of $\{\overline K,\overline H\}$ has the term
$c\,u_1^sM$, and all its other terms have smaller $u_1$-exponent.
The Jacobi identity and $\{\overline H,\overline J\}=0$ preserve the pair
equation. Since $K$ is derivative-only, specialness is preserved as in
Subsection~\ref{sec:symbols}.
Repeating the change and taking canonical representatives terminates and
yields a pair $(F_g,G_g)\in\mathscr K_{4g}$, with $\eta_g=uF_g+G_g$.
Each generator $K$ used in these canonical changes has KdV weight
$4g-3$. Since every derivative-only quadratic canonical monomial
is of the form $u_j^2$, of even KdV weight $2j+4$, the canonical
representative of $K$ has no quadratic part. Standardization therefore leaves the quadratic
coefficient unchanged.

\subsubsection{The heat product and its quadratic coefficient}\label{sec:heat-product}
We express the canonical first component through the product $W(T)=U(-T)U(T)$, as in \eqref{eq:heat-first-component}, and compute its quadratic coefficient. Since standardization preserves this coefficient, the calculation gives $c_g$ in \eqref{eq:algebraic-cA} and proves that the constructed pair is nonzero.

Let $\mathcal P(T)=\sum_{i\ge-2}P_iT^{i+2}$. The flux generating series is
\begin{equation}\label{eq:heat-flux}
 \mathcal P(T)=2\partial_TU(T)-\frac{U(T)}T-uU(T)
 -4\Dx^2U(T)+\frac1T.
\end{equation}
To obtain this formula, take the variational derivative of
$\{\overline{h_{i+1}^{\mathrm{KdV}}},\overline H\}=0$, using the identity
$\delta\overline{h_{i+1}^{\mathrm{KdV}}}/\delta u=h_i^{\mathrm{KdV}}$
from the preceding subsection. Together with the defining equation
$\partial_xP_i=\mathcal E_{Q_H}(h_i^{\mathrm{KdV}})$ and the recursion
\eqref{eq:heat-recursion}, this gives, with the constant fixed by homogeneity,
\[
 P_i=(2i+5)h_{i+1}^{\mathrm{KdV}}
     -u h_i^{\mathrm{KdV}}-4\partial_x^2h_i^{\mathrm{KdV}}.
\]
Summing against $T^{i+2}$ for $i\ge-2$ gives \eqref{eq:heat-flux};
the term $1/T$ comes from $h_{-2}^{\mathrm{KdV}}=1$ when shifting the index.
By \eqref{eq:local-pair-construction}, we have
\[
F_g^0 = [T^{2g-1}]U(-T)  \mathcal P(T) - h^{\mathrm{KdV}}_{2g - 2}.
\]
Define
\[
 W(T)=U(-T)U(T).
\]
Since $\partial_uU(T)=TU(T)$, the product $W(T)$ is derivative-only as an actual formal density. The odd part of $2U(-T)\partial_TU(T)$ is $\partial_TW(T)$. The terms $uW(T)$ and $U(-T)\Dx^2U(T)$ are even modulo total derivatives. The contribution from $U(-T)/T$ cancels the last term $h_{n+3}^{\mathrm{KdV}}=h_{2g-2}^{\mathrm{KdV}}$ in \eqref{eq:local-pair-construction}. Hence
\begin{equation}\label{eq:heat-first-component}
 \Can(F_g^0)=(2g-1)\Can\bigl([T^{2g}]W(T)\bigr),
\end{equation}
where $\Can$ denotes the canonicalization.
Here is a direct calculation of its quadratic coefficient. Set
\[
 \phi(z)=\sum_{j\ge0}\frac{j!}{(2j+1)!}z^j.
\]
The linear heat term, i.e., the term with monomial length $1$ in $U(T)$, is $T\phi(12T\Dx^2)u$. The quadratic term in its integral is
$\frac{T^2}{2}\int u\phi(12T\Dx^2)u\dd x$; this follows from the trace-variation identity above. Thus the quadratic part of $\int W(T)\dd x$ is
\[
 \int u\,T^2\left(\frac{\phi(z)+\phi(-z)}2-\phi(z)\phi(-z)\right)u\dd x,
 \qquad z=12T\Dx^2.
\]
The coefficient recursion for $\phi$ gives
$4z\phi'(z)+(2-z)\phi(z)=2$. Consequently, if $\psi(z)=\phi(z)\phi(-z)$, then
\[
 z\psi'(z)+\psi(z)=\frac{\phi(z)+\phi(-z)}2.
\]
For even $m=2g-2$, the coefficient of $z^m$ in $\left(\frac{\phi(z)+\phi(-z)}2-\phi(z)\phi(-z)\right)$ is
$\frac{m}{m+1}\frac{m!}{(2m+1)!}$. Replacing in \eqref{eq:heat-first-component}, followed by integration by parts, gives
\[
 [u_{2g - 2}^2]F_g=\frac{(2g - 2)12^{2g - 2} (2g - 2)!}{(2(2g - 2)+1)!}
 =\frac{(2g - 2)6^{2g - 2}}{(4g-3)!!},
\]
which is the first formula in \eqref{eq:algebraic-cA}.

\subsubsection{Standard diagonal extraction}\label{sec:diagonal-extraction}
We construct a linear functional $\lambda_g$ that extracts $[u_2^g]F_g$ directly from the heat product, before standardization. We prove its invariance under canonicalization and standardization to obtain \eqref{eq:diagonal-functional}. First we establish the length bound
\begin{equation}\label{eq:heat-length-bound}
 [T^{2g}]W(T)\text{ has monomial length at most }g.
\end{equation}
To prove it, use the Brownian-bridge form of the Feynman--Kac formula
\cite[Theorem 1.1 and equation (1.1)]{SimonFK}, with the time and potential
in that reference rescaled as $t=24T$ and $V=-u/24$. On the diagonal,
its normalized formal expansion is
\[
 U(T,x)=\mathbb E\exp\left(
 T\int_0^1u(x+\sqrt{24T}\,B(s))\dd s\right),
\]
where $B$ is a centered Gaussian bridge with covariance
$C(s,t)=\mathbb E[B(s)B(t)]=\min(s,t)-st$.
Here the expectation is evaluated coefficientwise: Taylor-expand
\[
 u(x+\sqrt{24T}\,B(s))
 =\sum_{j\ge0}\frac{u_j(x)}{j!}(24T)^{j/2}B(s)^j,
\]
then expand the exponential and evaluate each Gaussian moment by
\emph{Wick's pairing rule} \cite[Theorem 3.1]{Etingof}. Odd moments vanish, and
\[
 \mathbb E[B(s_1)\cdots B(s_{2r})]
 =\sum_{\pi}\prod_{\{a,b\}\in\pi}C(s_a,s_b),
\]
where $\pi$ runs over all partitions of $\{1,\ldots,2r\}$ into pairs.
Each pair is called a \emph{contraction}; the time variables are then
integrated over $[0,1]$.

Represent each factor $u_j$ by a vertex and each contraction by an edge, allowing an edge to join a vertex to itself. A graph with $\ell$ vertices and $E$ edges has monomial length $\ell$ and contributes $T^{\ell+E}$: the exponential supplies $T^\ell$, and the $2E$ contracted Gaussian factors supply $T^E$.

Contributions factor over connected components, so the exponential
formula for graphs says that $\log U$ contains precisely the connected
graphs \cite[Theorem 3.11]{Etingof}. A connected graph on $\ell$ vertices has at least $\ell-1$
edges. Therefore a length-$\ell$ term in $\log U$ has \emph{heat degree}, i.e., exponent of the formal parameter $T$,
at least $2\ell-1$. Since
\[
 \log W(T)=\log U(T)+\log U(-T),
\]
all odd heat degrees cancel; its length-$\ell$ terms thus have heat
degree at least $2\ell$. Finally, in $W=\exp(\log W)$, both lengths
and heat degrees add under multiplication. Every length-$\ell$ term
of $W$ therefore still has heat degree at least $2\ell$, which gives
$\ell\le g$ at degree $T^{2g}$ and proves \eqref{eq:heat-length-bound}.
All these expansions are interpreted coefficientwise in differential
polynomials.

Fix $g\ge2$. On derivative-only monomials of length $g$ and differential weight $2g$, define
\begin{equation}\label{eq:lambda-definition}
 \lambda_g\left(\prod_{j\ge1}u_j^{m_j}\right)
 =C_r\prod_{j\ge2}\bigl((2j-3)!!\bigr)^{m_j},\qquad
 C_r=\frac{(-1)^r r!}{(3g-3)(3g-4)\cdots(3g-r-2)},
 \quad r=m_1.
\end{equation}
An empty product is one. The denominators do not vanish, since $r\le g-1$. Extend $\lambda_g$ linearly to derivative-only polynomials, setting it equal to zero on every component whose length is not $g$ or whose differential weight is not $2g$.

This functional vanishes on total derivatives. Indeed, for a source monomial of length $g$, differential weight $2g-1$, and $u_1$-exponent $r$, its value on the derivative, after removing the common double-factorial product, is
\[
 rC_{r-1}+(3g-2-r)C_r=0.
\]
Therefore, $\lambda_g$ is invariant under canonicalization, and $\lambda_g (\Can (F^0_g)) = \lambda_g (F^0_g)$.
It also vanishes on density representatives of the Poisson bracket
$\{\overline K,\int u^3/6\dd x\}$ for every derivative-only density $K$
of length $g-1$ and differential weight $2g-1$. Integration by parts gives
\[
 \left\{\overline K,\int\frac{u^3}{6}\dd x\right\}
 =\int\left(
 \sum_{j\ge1}\frac{\partial K}{\partial u_j}
       \sum_{a=1}^{j}\binom ja u_a u_{j+1-a}-u_1K
 \right)\dd x.
\]
For a monomial $K=\prod_j u_j^{m_j}$ with $r=m_1$, the terms with $a=1$ or $a=j$, together with $-u_1K$, produce a new factor $u_1$. Their coefficient is
\[
 r+\sum_{j\ge2}(j+1)m_j-1=3g-3-r.
\]
After applying $\lambda_g$ and removing the common factor $\prod_{j\ge2}((2j-3)!!)^{m_j}$, they contribute $(3g-3-r)C_{r+1}$. The interior terms $2\le a\le j-1$ preserve the $u_1$-exponent, and their double-factorial convolution is
\[
 \sum_{a=2}^{j-1}\binom ja(2a-3)!!(2j-2a-1)!!
 =(j-2)(2j-3)!!,\qquad j\ge2.
\]
For example, this follows by expanding
$A(z)A'(z)=zA'(z)-2A(z)$ with
$A(z)=1-z-\sqrt{1-2z}$. Since
$\sum_{j\ge2}(j-2)m_j=(2g-1)-2(g-1)+r=r+1$,
the interior terms contribute $(r+1)C_r$ after removal of the same common factor. The value of the bracket is therefore proportional to
\[
 (3g-3-r)C_{r+1}+(r+1)C_r=0.
\]

In the standardization of a density with length at most $g$, every chosen generator has length at most $g-1$: it is obtained by dividing a cancelled monomial by $u_1$. In the decomposition
\[
 \{\overline K,\overline H\}
 =\left\{\overline K,\int\frac{u^3}{6}\dd x\right\}
  -\left\{\overline K,\int u_1^2\dd x\right\},
\]
the second bracket has length at most $g-1$, so only the first can affect
length $g$. These changes also preserve the length bound, since the first bracket has length at most $g$ and the second has length at most $g-1$. Therefore, $\lambda_g$ is invariant under standardization and $\lambda_g(F_g)=\lambda_g(F_g^0)$.

In the standard density $F_g$, a monomial of length $g$ and KdV weight $4g$ has derivative indices at least two with sum $2g$. All of them must therefore equal two, so this component is $[u_2^g]F_g\cdot u_2^g$. Since $\lambda_g(u_2^g)=1$ and $\lambda_g$ vanishes on the other components, \eqref{eq:heat-first-component} and the preceding invariance identities give
\begin{equation}\label{eq:diagonal-functional}
 [u_2^g]F_g=(2g-1)\lambda_g\bigl([T^{2g}]W(T)\bigr).
\end{equation}

\subsubsection{Evaluation by the classical heat-kernel amplitude}\label{sec:heat-amplitude}
We evaluate $\lambda_g([T^{2g}]W(T))$ to obtain the diagonal coefficient $A_g$ in \eqref{eq:algebraic-cA}. First we express the length-$g$ terms by the classical heat-kernel amplitude. We then realize $\lambda_g$ as an integral on a chosen potential and evaluate that integral.

Introduce an auxiliary parameter
$\hbar$ and denote by $\mathcal K_\hbar(T;x,y)$ the heat kernel defined by
\[
 \hbar\partial_T\mathcal K_\hbar(T;x,y)
 =(12\hbar^2\partial_x^2+u(x))\mathcal K_\hbar(T;x,y),
 \qquad \mathcal K_\hbar(0;x,y)=\delta(x-y).
\]
Here $T$ is the heat-time variable, $y$ is the initial spatial variable,
and the differential operator acts on $x$. We use only the formal local
expansion of this kernel. Its normalized diagonal is
\[
 U_\hbar(T,x):=\sqrt{48\pi\hbar T}\,\mathcal K_\hbar(T;x,x)
 =U(\hbar T,x;u/\hbar^2),
\]
where the last expression means the series \eqref{eq:heat-U} with
$T$ replaced by $\hbar T$ and every jet $u_j$ replaced by
$u_j/\hbar^2$. Thus a term of heat degree $n$ and monomial length
$\ell$ acquires the factor $\hbar^{n-2\ell}$. By
\eqref{eq:heat-length-bound}, the product
$U_\hbar(-T,x)U_\hbar(T,x)$ has no negative powers of $\hbar$.
At heat degree $2g$, its constant term in $\hbar$ is exactly the
length-$g$ component of $[T^{2g}]W(T)$. The substitution $T\mapsto-T$
is made in the formal series.

To compute this constant term, use the formal heat-equation version of the semiclassical action--amplitude expansion; see \cite{DWM} for the classical action, Jacobi fields, and Van Vleck amplitude. In our normalization it takes the form
\[
 \mathcal K_\hbar(T;x,y)\sim(48\pi\hbar T)^{-1/2}
 e^{S(T;x,y)/\hbar}\mathcal B(T;x,y)(1+O(\hbar)).
\]
We call $S(T;x,y)$ the \emph{classical action} and $\mathcal B(T;x,y)$ the \emph{leading amplitude}; for $u=0$ their normalization is
$S(T;x,y)=-(x-y)^2/(48T)$ and $\mathcal B(T;x,y)=1$.
Substitution into the defining heat equation and comparison of the first
two orders in $\hbar$ give the Hamilton--Jacobi and transport equations
\begin{align*}
 \partial_TS&=12(\partial_xS)^2+u(x),\\
 \partial_T\mathcal B-24(\partial_xS)(\partial_x\mathcal B)
 &=\left(12\partial_x^2S+\frac1{2T}\right)\mathcal B.
\end{align*}
The term $1/(2T)$ in the second equation comes from differentiating
the prefactor $T^{-1/2}$.

To solve the transport equation, consider the trajectories
$\gamma(\tau;y,v)$ with fixed initial position $y$ and initial velocity $v$:
\[
 \ddot\gamma=-24u'(\gamma),\qquad
 \gamma(0;y,v)=y,\quad \dot\gamma(0;y,v)=v.
\]
For the trajectory with $\gamma(T;y,v)=x$, the action is the stationary value
\[
 S(T;x,y)=\int_0^T
 \left(u(\gamma(\tau))-\frac{\dot\gamma(\tau)^2}{48}\right)\dd\tau.
\]
Endpoint variation gives
$\partial_xS(\tau;\gamma(\tau),y)=-\dot\gamma(\tau)/24$.
Thus these trajectories are the characteristics of the transport equation.
Define their variation with respect to the initial velocity by
\[
 j(\tau)=\frac{\partial\gamma(\tau;y,v)}{\partial v},\qquad
 \ddot j=-24u''(\gamma)j,\qquad j(0)=0,\quad\dot j(0)=1.
\]
Differentiating the endpoint identity with respect to $v$, with $\tau,y$
fixed, gives
\[
 (\partial_x^2S)(\tau;\gamma(\tau),y)\,j(\tau)
 =-\frac{\dot j(\tau)}{24}.
\]
Consequently the transport equation along the trajectory becomes
\[
 \frac{d}{d\tau}\log\mathcal B(\tau;\gamma(\tau),y)
 =\frac1{2\tau}-\frac{\dot j(\tau)}{2j(\tau)}
 =\frac12\frac{d}{d\tau}\log\frac{\tau}{j(\tau)}.
\]
Integrating and using $j(\tau)\sim\tau$ and
$\mathcal B(\tau;\gamma(\tau),y)\to1$ as $\tau\to0$ fixes the
multiplicative constant. Taking $y=x$ and the initial velocity for which
$\gamma(T;x,v)=x$ therefore yields
\begin{equation}\label{eq:van-vleck-normalization}
 \mathcal B(T;x,x)^2=\frac{T}{j(T)}.
\end{equation}
Here the derivative defining $j$ keeps the initial position fixed;
nearby trajectories obtained by varying $v$ need not return to $x$.

To determine the parity in $T$, compare this expansion with the connected-graph expansion in Subsection~\ref{sec:diagonal-extraction}. On the diagonal,
\[
 \log U_\hbar(T,x)=\hbar^{-1}S(T;x,x)
                  +\log\mathcal B(T;x,x)+O(\hbar).
\]
A connected graph with $\ell$ vertices and $E$ edges contributes $T^{\ell+E}\hbar^{E-\ell}$. Thus the coefficient of $\hbar^{-1}$ consists of trees ($E=\ell-1$), and the coefficient of $\hbar^0$ consists of one-loop graphs ($E=\ell$). Comparing coefficients identifies them with $S$ and $\log\mathcal B$, respectively. This is the usual tree and one-loop expansion \cite[Section 3.6, Theorem 3.12]{Etingof}, applied coefficientwise to the Brownian-bridge contractions, with their time variables integrated.

The corresponding powers $T^{2\ell-1}$ and $T^{2\ell}$ show that the diagonal action is odd in $T$ and $\log\mathcal B$, hence $\mathcal B$, is even. The actions therefore cancel in the normalized product at opposite times, giving
\[
 [\hbar^0]\bigl(U_\hbar(-T,x)U_\hbar(T,x)\bigr)
 =\mathcal B(T;x,x)^2=\frac{T}{j(T)}.
\]
This is the generating series for the required length-$g$ components
at heat degree $2g$.

Consider the potential defined in terms of a parameter $k\le0$ by
\begin{equation}\label{eq:sheared-parabola}
 x=k-\frac{k^2}{2},\qquad u(x)=\frac{k^2}{2}.
\end{equation}
It has derivatives
\[
 u_1=\frac{k}{1-k},\qquad
 u_j=\frac{(2j-3)!!}{(1-k)^{2j-1}}\quad(j\ge2).
\]
For a derivative-only polynomial $P$ of length $g$ and differential weight $2g$, \eqref{eq:lambda-definition} therefore has the convergent integral representation
\begin{equation}\label{eq:lambda-beta}
 \lambda_g(P)=(3g-2)\int_{-\infty}^{0}(1-k)
 \left.P\right|_{\substack{u_1=k/(1-k)\\
 u_j=(2j-3)!!/(1-k)^{2j-1}\ (j\ge2)}}\dd k.
\end{equation}
For a monomial with $r$ copies of $u_1$, this is exactly the beta integral
\[
 (3g-2)\int_{-\infty}^{0}k^r(1-k)^{1-3g}\dd k=C_r.
\]

The classical trajectory in this potential can be integrated explicitly.
Since we evaluate the heat kernel on the diagonal, we consider a trajectory
satisfying $\gamma(0)=\gamma(T)=x$. Along it, write
\[
 \gamma(\tau)=v(\tau)-\frac{v(\tau)^2}{2},\qquad
 u(\gamma(\tau))=\frac{v(\tau)^2}{2},\qquad v(0)=v(T)=k.
\]
Here $v(\tau)$ varies along the trajectory, whereas $k$ is its fixed value at the common initial and final position. Let the turning point be $v=b$, and put $v=b\cos\varphi$. Conservation of energy gives
$\dot\gamma(\tau)^2=24\bigl(b^2-v(\tau)^2\bigr)$. For the symmetric trajectory from the endpoint $v=k$ back to itself, write $k=b\cos\theta$. Half of its travel time is
\[
 \frac T2=\frac1{\sqrt{24}}\int_0^\theta(1-b\cos\varphi)\dd\varphi.
\]
With $z=\sqrt{24}T/2$, this becomes
\begin{equation}\label{eq:classical-angle}
 z=\theta-k\tan\theta,\qquad
 \theta=\frac{z}{1-k}+O(z^3).
\end{equation}
To compute the Jacobi field in \eqref{eq:van-vleck-normalization}, set $E=b^2/2$ and differentiate half of the travel time at fixed endpoint $k$. It gives
\[
 \partial_E(T/2)=\frac{\cos\theta-b}{\sqrt{24}\,b^2\sin\theta},
 \qquad \dot\gamma(0)=\sqrt{24}\,b\sin\theta.
\]
To relate this derivative to $j(T)$, write $w=\dot\gamma(0)$ for the initial velocity and $R(w)$ for the return time at fixed initial position $x$. Differentiating $\gamma(R(w);x,w)=x$ gives
\[
 j(T)+\dot\gamma(T)R'(w)=0,\qquad T=R(w).
\]
Energy conservation gives $\dot\gamma(T)=-w$ and $E=u(x)+w^2/48$, so
\[
 j(T)=wR'(w)=\frac{w^2}{24}\frac{\partial T}{\partial E},
\]
where $T=T(E;k)$ is the return time as a function of the conserved energy $E$, and the derivative is taken with the endpoint parameter $k$ held fixed.
Substituting the initial velocity $w=\dot\gamma(0)=\sqrt{24}\,b\sin\theta$,
the displayed derivative of $T/2$, and $b=k/\cos\theta$ yields
\[
 j(T)=\frac{2}{\sqrt{24}}
       (\sin\theta\cos\theta-k\tan\theta).
\]
The formulas extend by their formal limits when $k=0$. Thus, for the potential \eqref{eq:sheared-parabola}, the generating
series in $T$ of the length-$g$ components of $[T^{2g}]W(T)$,
written in terms of $z=\sqrt{24}\,T/2$, is
\begin{equation}\label{eq:classical-J}
 \mathcal J(z,k)=\frac{z}{\sin\theta\cos\theta-k\tan\theta}.
\end{equation}
At $k=0$ this is $2z/\sin(2z)$, also obtained directly from the quadratic heat kernel.

It remains to evaluate the following integral for small $z>0$ and
extract its Taylor coefficients at $z=0$:
\[
 I(z)=\int_{-\infty}^0(1-k)(\mathcal J(z,k)-1)\dd k.
\]
For small positive $z$, \eqref{eq:classical-angle} gives
\[
 k=(\theta-z)\cot\theta,\qquad
 \frac{\dd k}{\dd\theta}
 =\frac{z-\theta+\sin\theta\cos\theta}{\sin^2\theta},
 \qquad \mathcal J\dd k=z\csc^2\theta\dd\theta.
\]
The endpoints become $\theta=0,z$. A primitive of the resulting integrand is
\[
 -\frac z2\cot\theta-
 \frac z2(z-\theta)\csc^2\theta-k+\frac{k^2}{2}.
\]
Its upper value is $-z\cot z/2$, and its lower limit is $-1/2-z^2/2$. Hence
\begin{equation}\label{eq:scalar-integral}
 I(z)=\frac{1-z\cot z+z^2}{2}.
\end{equation}
To recover $\lambda_g$ from $I(z)$, we use
\[
 [z^{2g}]I(z)=\int_{-\infty}^{0}(1-k)
 [z^{2g}]\mathcal J(z,k)\dd k,\qquad g\ge1.
\]
To justify this extraction, set $\alpha=(1-k)^{-1}$.
The implicit relation \eqref{eq:classical-angle} gives
$\mathcal J-1=\alpha^3z^2F(z,\alpha)$, where $F$ is analytic in $z$
uniformly for $\alpha\in[0,1]$. Since
$(1-k)\dd k=\alpha^{-3}\dd\alpha$, we have
$I(z)=z^2\int_0^1F(z,\alpha)\dd\alpha$, which may be expanded termwise.

Recall that $1-(t/2)\cot(t/2)=\sum_{g\ge1}|B_{2g}|t^{2g}/(2g)!$. For $g\ge2$, the term $z^2/2$ does not contribute to the coefficient
of $z^{2g}$ in \eqref{eq:scalar-integral}. Equations \eqref{eq:lambda-beta}--\eqref{eq:scalar-integral}, with $z=\sqrt{24}T/2$, give
\[
 \lambda_g\bigl([T^{2g}]W(T)\bigr)
 =\frac{3g-2}{2}\,\frac{24^g|B_{2g}|}{(2g)!}.
\]
Together with \eqref{eq:diagonal-functional}, this proves the second formula in \eqref{eq:algebraic-cA}.
Thus both coefficients in \eqref{eq:algebraic-cA} are nonzero. Proposition~\ref{prop:quad} now shows that the constructed pair spans $\mathscr K_{4g}$; dividing its two coefficients gives \eqref{eq:diagonal-ratio} and proves Proposition~\ref{prop:diagonal-algebraic}.
\end{proof}

\paragraph{Completion of the proof of Theorem~\ref{thm:pair}.}
Step 1 proves injectivity of the first-component projection. Proposition~\ref{prop:quad} proves vanishing in odd weights and the dimension bound in even weights (Steps 2 and 3). Proposition~\ref{prop:quartic-vanishing} excludes $d\equiv2\pmod4$ (Step 4). Finally, Proposition~\ref{prop:diagonal-algebraic} supplies a nonzero pair at every weight $d=4g$, so the dimension bound is attained; its nonzero diagonal coefficient proves the remaining assertion (Step 5). This proves Theorem~\ref{thm:pair}.

\section{Proof of Theorem~\ref{thm:generic}}\label{sec:nonlinear}
We prove Theorem~\ref{thm:generic}. Normalize both hierarchies to $a=-1$ by the common rescaling of Section~\ref{sec:normalization}. Write $H_p^{[r]}$ for the component of coefficient degree $r$, hence of KdV weight $2(p+2)+r$. Generalized standard form gives $H_1=H+F_{\mathrm{positive}}$, where every monomial of $F_{\mathrm{positive}}$ has KdV weight greater than six, or equivalently positive coefficient degree.

A variational derivative lowers KdV weight by two and $\partial_x$ raises it by one. Thus a bracket of first- and second-Hamiltonian components of coefficient degrees $r,s$ has KdV weight
\[
 (6+r)+(8+s)-3=11+r+s.
\]
We use this weight to extract homogeneous components of the commutativity equations below.

By specialness, the second Hamiltonian has canonical density
\[
 H_2=\frac{u^4}{24}-u u_1^2+uF_{\mathrm{positive}}+G_{\mathrm{derivative}},
\]
with $G_{\mathrm{derivative}}$ derivative-only. Each nonzero monomial in $G_{\mathrm{derivative}}$ has at least two factors (because the density is canonical), all of positive derivative index. Hence $L\ge2$ and $N\ge L$, so its coefficient degree satisfies
\[
 r=N+2L-8\ge3L-8\ge-2.
\]
The other displayed terms have degree zero or positive degree. Therefore $H_2$ has no components below degree $-2$.

Suppose $\nu<0$ is the lowest degree with $H_2^{[\nu]}\ne0$. This component is derivative-only. Extracting KdV weight $11+\nu$ from $\{\overline H_1,\overline H_2\}=0$ gives
\[
 0=\sum_{\substack{r\ge0,\ s\ge-2\\r+s=\nu}}
       \{\overline H_1^{[r]},\overline H_2^{[s]}\}
   =\{\overline H,\overline H_2^{[\nu]}\}.
\]
Indeed, any $r>0$ would require $s=\nu-r<\nu$, whose component vanishes by the choice of $\nu$. Lemma~\ref{lem:topv} now gives $H_2^{[\nu]}=0$, a contradiction. Thus $H_2$ has no negative-degree components. In degree zero, commutativity and specialness give
\[
 \{\overline H,\overline H_2^{[0]}\}=0,
 \qquad \partial_u\overline H_2^{[0]}=\overline H.
\]
The density $J$ in \eqref{eq:base} satisfies both equations. Their difference $H_2^{[0]}-J$ is derivative-only and commutes with $\overline H$, so Lemma~\ref{lem:topv} gives $H_2^{[0]}=J$. Both hierarchies therefore have the same degree-zero KdV pair $(H,J)$ and only nonnegative coefficient degrees.

Denote their differences by $\Delta\overline H_p=\overline{\widetilde H}_p-\overline H_p$. Subtracting specialness and commutativity gives the exact identities
\[
 \partial_u\Delta\overline H_2=\Delta\overline H_1,
\]
\[
 \{\overline H_1,\Delta\overline H_2\}
 +\{\Delta\overline H_1,\overline H_2\}
 +\{\Delta\overline H_1,\Delta\overline H_2\}=0.
\]
Here $\partial_u$ lowers KdV weight by two, so the specialness identity relates components of the same coefficient degree in $H_2$ and $H_1$.

Suppose the first Hamiltonians differ, and let $\mu>0$ be their first discrepancy degree. If $\nu<\mu$ were the first discrepancy degree of the second Hamiltonians, these identities would give
\[
 \partial_u\Delta\overline H_2^{[\nu]}
   =\Delta\overline H_1^{[\nu]}=0,
 \qquad
 \{\overline H,\Delta\overline H_2^{[\nu]}\}=0.
\]
For the second equation, extract KdV weight $11+\nu$: every term involving $\Delta\overline H_1$ has higher weight, and a positive-degree component of $H_1$ would require a component of $\Delta\overline H_2$ below $\nu$. The first equation makes $\Delta\overline H_2^{[\nu]}$ derivative-only; the second makes it a KdV centralizer. Lemma~\ref{lem:topv} therefore excludes such a $\nu$.

Consequently both differences have no components below coefficient degree $\mu$. Write their degree-$\mu$ canonical densities as $F$ and $\eta$. Specialness gives
\[
 \partial_u\overline\eta=\overline F,
 \qquad \eta=uF+G,
\]
where $G$ is derivative-only. Since $F\ne0$ by the choice of $\mu$, also $\eta\ne0$: the first discrepancies in the two Hamiltonians occur at the same coefficient degree, with respective KdV weights $\mu+6$ and $\mu+8$.

Every bracket in the subtracted commutativity equation contains a discrepancy of degree at least $\mu$ and another term of nonnegative degree. Thus all its components of KdV weight below $11+\mu$ vanish identically. At weight $11+\mu$, only a degree-zero background paired with a degree-$\mu$ discrepancy can contribute. A positive-degree background gives greater weight, and the bracket of two discrepancies has weight at least $11+2\mu>11+\mu$. Extracting precisely weight $11+\mu$ therefore gives
\[
 \{\overline H,\overline{uF+G}\}+\{\overline F,\overline J\}=0.
\]
Since $F\in V_{\mu+6}$ and $G\in W_{\mu+8}$, this is exactly the defining equation of the pair kernel:
\[
 (F,G)\in\mathscr K_d,\qquad d=\mu+6.
\]

We now apply Theorem~\ref{thm:pair} to this discrepancy pair. For $\mu\ge2$, we have $d\ge8$, as required by the theorem. If $d\not\equiv0\pmod4$, it gives $\mathscr K_d=0$. If $d=4g$, it gives a spanning pair $(F_g,G_g)$ with $[u_2^g]F_g\ne0$. Writing $(F,G)=c(F_g,G_g)$, equality of the diagonal data yields
\[
 0=[u_2^g]F=c\,[u_2^g]F_g,
\]
so $c=0$ and again $F=G=0$. The remaining case $\mu=1$, outside the range of Theorem~\ref{thm:pair}, follows from $V_7=0$ and Lemma~\ref{lem:topv}. Thus the discrepancy vanishes in every case, contradicting the choice of $\mu$.

The first Hamiltonians therefore agree. Liu--Zhang's uniqueness result \cite[Lemma 3.3 and Theorem 3.5]{LZ} identifies the Hamiltonian vector fields of the higher Hamiltonians: they commute with the same first field $uu_1+O(\eps)$ and have the same prescribed dispersionless parts. The difference of the corresponding Hamiltonians is therefore a Casimir of $\partial_x$, hence a multiple of $\int u\dd x$. The common dispersionless part and the positive differential weights of all higher corrections exclude this difference. This is the local-functional uniqueness statement recalled in \cite[Lemma 2.1(2)]{BR}. Undoing the common rescaling proves the theorem.

\end{document}